\documentclass[en,12pt]{article}
\usepackage{amsmath}
\usepackage{graphicx,psfrag,epsf}
\usepackage{enumerate}
\usepackage{natbib}
\usepackage{url} 
\usepackage{pdfpages}

\newcommand{\blind}{1}

\def\spacingset#1{\renewcommand{\baselinestretch}%
{#1}\small\normalsize} \spacingset{1}

\usepackage{xr}
\usepackage{mathrsfs}
\usepackage{amsthm}
\usepackage{amssymb}
\usepackage{multirow}
\usepackage{xargs}[2008/03/08]
\usepackage{mathtools}
\allowdisplaybreaks

\usepackage{xcolor}
\usepackage{enumitem}
\usepackage{hyperref}
\hypersetup{
    colorlinks,
    linkcolor={red!50!black},
    citecolor={blue!50!black},
    urlcolor={blue!80!black}
}

\theoremstyle{plain}
\newtheorem{theorem}{Theorem}
  \theoremstyle{plain}
  \newtheorem{lemma}[theorem]{Lemma}
  \theoremstyle{plain}
  
  \newtheorem{proposition}[theorem]{Proposition}

\global\long\def\expect{\mathbb{E}}
\global\long\def\var{\mathrm{Var}}
\global\long\def\prob{\mathrm{Pr}}

\global\long\def\real{\mathbb{R}}
\global\long\def\Op{O_{P}}
\global\long\def\manifold{\mathcal{M}}

\global\long\def\op{o_{P}}

\global\long\def\covarop{\mathcal{C}}
\global\long\def\asymplt{\lesssim}

\global\long\def\asympeq{\asymp}

\global\long\def\diffop{\mathrm{d}}
\newcommandx\tangentspace[2][usedefault, addprefix=\global, 1=\manifold]{T_{#2}#1}

\newcommandx\lpnorm[3][usedefault, addprefix=\global, 1=r, 2=]{\|#3\|_{\mathcal{L}^{#1}}^{#2}}
\newcommandx\lp[1][usedefault, addprefix=\global, 1=p]{\mathcal{L}^{#1}}
\global\long\def\transpose{\mathrm{T}}

\global\long\def\ltwo{\mathcal{L}^{2}}

\newcommandx\vfnorm[3][usedefault, addprefix=\global, 1=\mu, 2=]{\|#3\|_{#1}^{#2}}
\newcommandx\vfinnerprod[2][usedefault, addprefix=\global, 1=\mu]{\llangle#2\rrangle_{#1}}

\global\long\def\tdomain{\mathcal{T}}

\newcommandx\opnorm[3][usedefault, addprefix=\global, 1=\mu, 2=]{\vertiii{#3}_{#1}^{#2}}
\newcommandx\fronorm[2][usedefault, addprefix=\global, 1=]{|#2|_{F}^{#1}}

\begin{document}


\if1\blind
{
  \title{\bf Mean and Covariance Estimation for Functional Snippets}
  \date{}
  \author{Zhenhua Lin\thanks{Email: stalz@nus.edu.sg. Department of Statistics and Applied Probability,  National University of Singapore. 
    Research supported by NIH ECHO grant (5UG3OD023313-03).}\hspace{.2cm} 
    and  
    Jane-Ling Wang\thanks{Email: janelwang@ucdavis.edu. Department of Statistics, University of California at Davis. Research supported by NIH ECHO grant (5UG3OD023313-03) and NSF (15-12975 and 19-14917).} 
    }
  \maketitle
} \fi


\begin{abstract}
We consider estimation of mean and  covariance functions of functional snippets, which are short segments of functions possibly observed  irregularly on an individual specific subinterval that is much shorter than the entire study  interval. Estimation of the covariance function for functional snippets is challenging  since information for the far off-diagonal regions of the covariance structure is completely missing. We address this difficulty by decomposing the covariance function into a variance function component and a correlation function component. The variance function can be effectively estimated nonparametrically, while the correlation part is modeled parametrically, possibly with an increasing number of parameters,  to handle the missing information in the far off-diagonal regions. Both  theoretical analysis and numerical simulations suggest that this hybrid  strategy 
is effective. In addition, we propose a new  estimator for the variance of measurement errors and analyze its asymptotic properties. This estimator is required for the estimation of the variance function from noisy measurements.
\end{abstract}

\noindent%
\vfill

\spacingset{1.45} 

\section{Introduction}
\label{sec:introduction}
Functional data are random functions on a common domain, e.g.,  an interval  $\tdomain\subset\real$. In reality they can only be observed  on a discrete schedule, possibly intermittently, which leads to an incomplete data problem. Luckily, by now  this problem has largely been resolved \citep{Rice2001, Yao2005a, Li2010, Zhang2016} and  there is  a large literature  
on the analysis of  functional  data.  For a more comprehensive treatment readers are referred to the monographs by \citet{Ramsay2005}, \citet{Ferraty2006},   \citet{Hsing2015} and \cite{Kokoszka2017}, and a review paper  by \cite{Wang2016}.

In this paper, we address a different type of incomplete data, which occurs frequently  in longitudinal studies when subjects enter the study at  random time and  are followed for a short period  within the  domain  $\tdomain=[a,b]\subset\real$. Specifically, we focus on functional data with the following property: each function $X_{i}$
is only observed on a subject-specific interval $O_{i}=[A_{i},B_{i}]\subset[a,b]$, and 
\begin{enumerate}
	[leftmargin=*,labelsep=7mm,label= (S)]
	\item \label{emu:s1} there exists an absolute constant $\delta$ such that $0<\delta<1$ and  $B_i-A_i \leq\delta (b-a)$ for all $i=1,2,\ldots$.
\end{enumerate}
As a result,  the design of support points \citep{Yao2005a} where one has information about the covariance function $\covarop (s, t)$ 
is incomplete in the sense that there are  no design points in the 
off-diagonal region, $\tdomain_\delta^c=\{(s, t) : \ \mid s - t \mid > \delta(b-a), s,t\in[a,b] \}$. This is mathematically characterized by 
\begin{equation}\label{eq:snippet}
\left(\bigcup_{i}[A_{i},B_{i}]^{2}\right)\cap\mathcal{T}_{\delta}^{c}=\emptyset.
\end{equation}
Consequently, local smoothing methods, such as  PACE \citep{Yao2005a}, that are interpolation methods fail to produce a consistent estimate of the covariance function in the off-diagonal region as the problem requires data extrapolation.

An example is the spinal bone mineral density data collected from 423 subjects ranging in  age from 8.8 to 26.2 years \citep{BACHRACH1999}.  The design plot for the covariance function, as shown in Figure \ref{fig:bone-design}, indicates  that all of the design points  fall within  a narrow band  around the diagonal area but the domain of interest $[8.8, 26.2]$ is much larger than this band. The cause of this phenomenon is that each individual trajectory is only recorded in an individual specific subinterval that is much shorter than the span of the study.  For the spinal bone mineral density data, the span (length of interval between the first measurement and the last one) for each individual is no larger than 4.3 years, while the span for the study is about 17 years.  Data with this characteristic, mathematically described by \ref{emu:s1} or \eqref{eq:snippet}, are called  \emph{functional snippets} in this paper, analogous to the longitudinal snippets studied in \cite{Dawson2018}.   As it turns out, functional snippets are quite common in longitudinal studies  \citep{Raudenbush1992,Galbraith2017} and require extrapolation methods to handle.  Usually, this is not an issue for parametric approaches, such as linear mixed-effects models, but  requires a thoughtful plan for non- and semi-parametric approaches.  

Functional  fragments  \citep{Liebl:2013aa, Kraus2015, Kraus2019, Kneip:2019,  Liebl:2019aa},  like functional snippets, are also partially observed functional data and have been studied broadly in the literature. 
However, for data investigated in these works as functional fragments,  the span of a single individual domain $[A_i,B_i]$ can be nearly as large as the span $[a,b]$ of the study, making them distinctively different from functional snippets.  Such data, collectively referred to as  ``nonsnippet functional data'' in this paper, often satisfy the following condition:
\begin{enumerate}
	[leftmargin=*,labelsep=7mm,label= (F)]
	\item\label{P1}  for any $\epsilon\in(0,1)$,  $\lim_n\prob\{B_{i_n}-A_{i_n}>(1-\epsilon)(b-a)\}>0$ for a strictly increasing sequence $\{i_n\}_{n=1}^\infty$.
\end{enumerate}
For instance, \cite{Kneip:2019} assumed that $\prob([A_i,B_i]^2=[a,b]^2)>0$, which implies that design points and local information are still available in the off-diagonal region $\tdomain_\delta^c$. In other words, for non-snippet functional data and for each $(s,t)\in[a,b]^{2}$, one has $\prob\{(s,t)\in\bigcup_{i=1}^{n}[A_i,B_i]^2\}>0$
for sufficiently large $n$, contrasting with \eqref{eq:snippet} for functional snippets. Other related works by  \cite{Gellar:2014aa,Goldberg:2014aa,Gromenko:2017aa,Stefanucci2018} on partially observed functional data, although do not explicitly discuss the design, require condition \ref{P1} for their proposed methodologies and theory. All of them can be handled with a proper interpolation method, which is fundamentally different from the extrapolation methods needed for functional snippets.


The analysis of functional snippets is  more challenging than non-snippet functional data, since   information in the far off-diagonal regions of the covariance structure is completely missing for functional snippets according to \eqref{eq:snippet}. \citet{Delaigle2016} addressed  this challenge  by assuming that the underlying  functional data  are Markov processes, which is only valid at the discrete level, as pointed out by \cite{Descary2019}.   
\citet{Zhang2018} and \citet{Descary2019} used  matrix completion methods to handle functional snippets, but their approaches   require modifications to handle longitudinally recorded
snippets that are sampled at random design points, and their theory does not  cover random designs. \cite{Delaigle2019} proposed to numerically extrapolate an  estimate, such as PACE \citep{Yao2005a}, from the diagonal region to the entire domain via basis expansion.   In this paper, we propose a divide-and-conquer strategy to analyze (longitudinal) functional snippets with a focus on the mean and covariance estimation.  Once the covariance function has been estimated, functional principal component analysis can be performed through the spectral decomposition of the covariance operator.  
\begin{figure}[t]
\label{fig:bone-design}
\begin{center}
\includegraphics[width=3in]{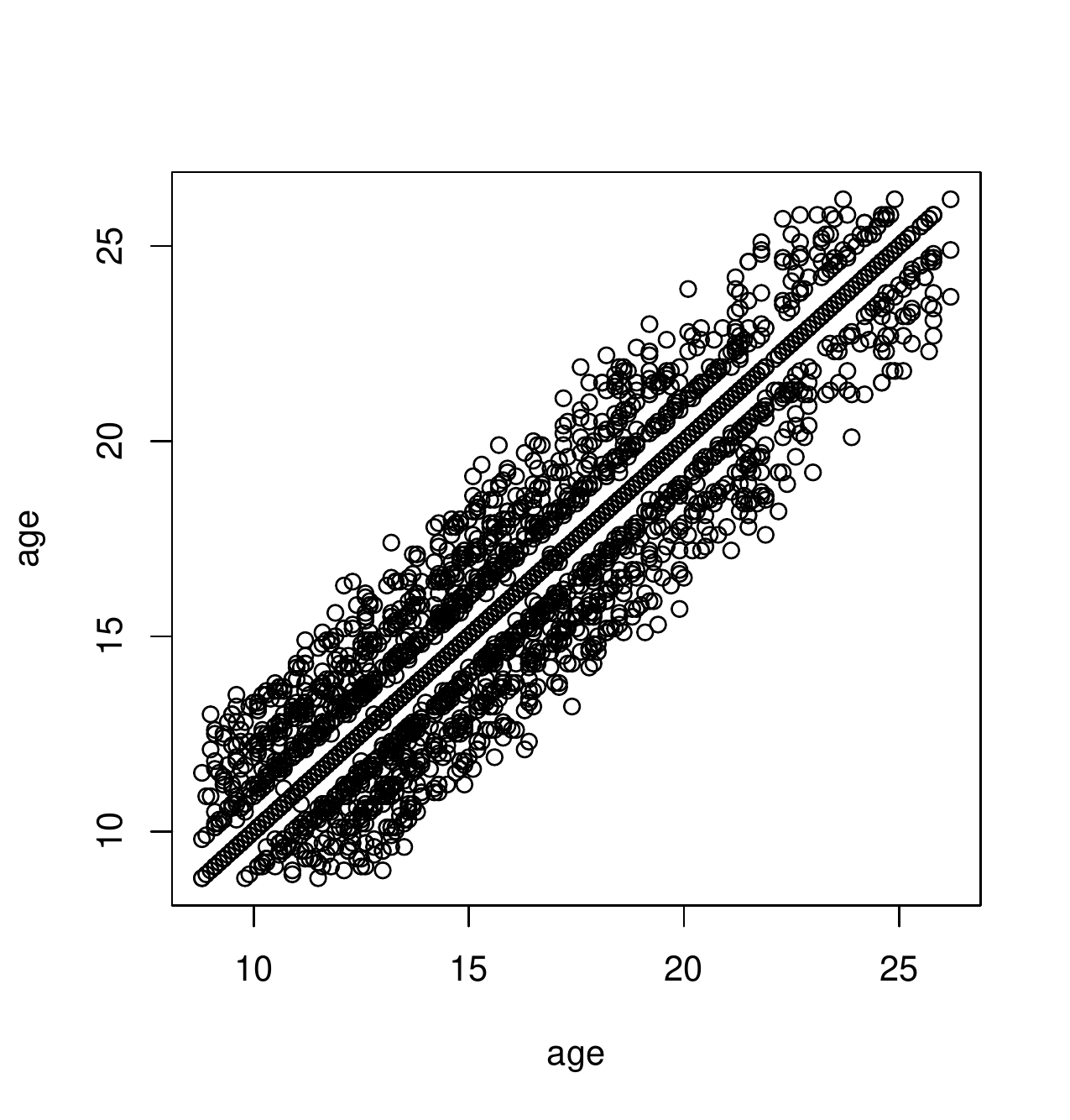}
\end{center}
\vspace{-1cm}
\caption{The design of covariance function from spinal bone mineral density data.}
\end{figure}

Specifically, we divide
the covariance function into two components, the variance function
and the correlation function. The former can be estimated via classic kernel smoothing, while the latter is modeled parametrically with a potentially diverging number of parameters. The principle behind this idea is to nonparametrically estimate the unknown components for which sufficient  information is available  while parameterizing the component with missing pieces. Since the correlation structure is usually much more structured than the covariance surface and it is possible to estimate the correlation structure nonparametrically within the diagonal band, a parametric correlation model  can be selected from  candidate models in existing literature and this usually works quite well to fit the unknown correlation structure. 

Compared to
the aforementioned works, our proposal enjoys at least two advantages. First, it can be applied to all types of designs, either sparsely/densely or regularly/irregularly observed snippets. 
Second, our approach is simple thanks to  the parametric structure of the correlation structure, and yet powerful due to the potential to accommodate  growing dimension of parameters and nonparametric variance component.  We stress that, our semi-parametric and divide-and-conquer strategy is fundamentally different from the penalized basis expansion approach that is adopted in the recent paper by \cite{Lin2019a} where the covariance function is represented by an analytic basis and the basis coefficients are estimated via penalized least squares. Numerical comparison of these two methods is  provided in Section \ref{sec:simulation}.

This divide-and-conquer approach has been explored in \cite{Fan2007} and \cite{Fan2008} to model the covariance structure of time-varying random noise in a varying-coefficient partially linear model.  We demonstrate here that a similar strategy can overcome the challenge  of the missing data issue  in functional snippets and further allow  the dimension of the correlation function to grow to infinity.   In addition, we take into account the measurement error in the observed data,  which is an important component in functional data analysis but is of less interest in a partially linear model and thus not considered  in \cite{Fan2007} and \cite{Fan2008}. The presence of measurement errors complicates the estimation of the variance function, as they are  entangled together along the diagonal direction of the covariance surface. Consequently, the estimation procedure for the variance function in \cite{Fan2007} and \cite{Fan2008} does not apply. 
While it is possible to estimate  the error variance using the approach in \citet{Yao2005a} and \citet{Liu2009},  these methods require a pilot estimate of the covariance function in the diagonal area, which involves two-dimensional smoothing, and thus are not efficient.  A key contribution of this paper is  a new estimator for the error variance in Section \ref{sec:sigma0} that is simple and easy to compute.  It improves upon the estimators in \citet{Yao2005a} and \citet{Liu2009}, as demonstrated through theoretical analysis and numerical studies; see Section \ref{sec:theory} and \ref{sec:simulation} for details.


\section{Mean and Covariance Function Estimation}\label{sec:method}

Let $X$ be a second-order random process defined on an interval $\tdomain\subset\real$
with mean function $\mu(t)= \expect X(t), $ and covariance function $\covarop(s,t)  = \mathrm{cov} (X(s), X(t)) $.
Without loss of generality, we  assume $\tdomain=[0,1]$ in the
sequel. 

Suppose $\{X_{1},\ldots,X_{n}\}$ is an independent random sample of $X$,
where $n$ is the sample size. In practice, functional data are rarely fully observed. Instead, they are often noisily recorded at some discrete points on $\tdomain$. To accommodate this practice, we  assume that each $X_i$ is only measured at $m_i$ points $T_{i1},\ldots,T_{im_i}$, and the observed data are $Y_{ij}=X_{i}(T_{ij})+\varepsilon_{ij}$ for $j=1,\ldots,m_i$, where  $\varepsilon_{ij}$ represents the  homoscedastic random
noise such that $\expect\varepsilon_{ij}=0$ and $\expect\varepsilon_{ij}^{2}=\sigma_{0}^{2}$. This homoscedasticity assumption  can be relaxed to accommodate heteroscedastic noise; see Section \ref{sec:sigma0} for details. 
To further elaborate the functional snippets characterized by \ref{emu:s1}, we assume  that the $i$th subject is only available to be studied between time $O_i-\delta/2$ and $O_i+\delta/2$, where the variable $O_i\in[\delta/2,1-\delta/2]$, called reference time in this paper, is specific to each subject and is modeled as identically and independently distributed (i.i.d.) random variables. We then assume that, $T_{i1},\ldots,T_{im_{i}}$ are i.i.d., conditional on $O_i$. These assumptions reflect the reality of many data collection processes  when subjects enter a study at  random time $O_i-\delta/2$ and are followed for a fixed period of time.  Such a sampling plan,  termed accelerated longitudinal design, has the advantage to expand the  time range of interest  in a short period of time as compared to a single cohort longitudinal design study. 

\subsection{Mean Function}

Even though only functional snippets are observed rather than a full curve, smoothing approaches such as \citet{Yao2005a} can be applied to estimate the mean function $\mu$, since for each $t$, there is positive probability that some design points fall into a small neighborhood of $t$. Here, we adopt a ridged version of the local linear
smoothing method in \citet{Zhang2016}, as follows. 

Let $K$ be a kernel function and $h_{\mu}$
a bandwidth, and define $K_{h_{\mu}}(u)=h_{\mu}^{-1}K(u/h_{\mu})$. The non-ridged local linear estimate of $\mu$ is given by $\tilde{\mu}(t)=\hat{b}_{0}$
with
\[
(\hat{b}_{0},\hat{b}_{1})=\underset{(b_{0},b_{1})\in\real^{2}}{\arg\min}\sum_{i=1}^{n}w_{i}\sum_{j=1}^{m_{i}}K_{h_{\mu}}(T_{ij}-t)\{Y_{ij}-b_{0}-b_{1}(T_{ij}-t)\}^{2},
\]
where $w_{i}\geq0$ are weight such that $\sum_{i=1}^{n}m_{i}w_{i}=1$. For the optimal choice of weight, readers
are referred to \citet{Zhang2018a}. It can be shown that
$\tilde{\mu}(t)=(R_{0}S_{2}-R_{1}S_{1})/(S_{0}S_{2}-S_{1}^{2})$,
where 
\begin{align*}
S_{r} & =\sum_{i=1}^{n}w_{i}\sum_{j=1}^{m_{i}}K_{h_{\mu}}(T_{ij}-t)\{(T_{ij}-t)/h_{\mu}\}^{r},\\
R_{r} & =\sum_{i=1}^{n}w_{i}\sum_{j=1}^{m_{i}}K_{h_{\mu}}(T_{ij}-t)\{(T_{ij}-t)/h_{\mu}\}^{r}Y_{ij}.
\end{align*}
Although $\tilde{\mu}$ behaves well most of the time, for a finite sample, there is positive probability that $S_0S_2-S_1^2=0$, hence  $\tilde{\mu}$ may become undefined. This minor issue can be addressed by ridging, a regularization technique used by \cite{Fan1993} with details in \cite{Seifert1996} and \cite{Hall1997}. The basic idea is to add a small positive constant to the denominator of $\tilde{\mu}$ when $S_0S_2-S_1^2$ falls below a threshold. More specifically, the ridged version of $\tilde{\mu}(t)$ is given by 
\begin{equation}\label{eq:ridged-mean}
\hat{\mu}(t)=\frac{R_{0}S_{2}-R_{1}S_{1}}{S_{0}S_{2}-S_{1}^{2}+\Delta1_{\{|S_{0}S_{2}-S_{1}^{2}|<\Delta\}}},
\end{equation}
where $\Delta$ is a sufficiently small constant depending on $n$
and $m_1,\ldots,m_n$.  A convenient choice here is $\Delta=(nm)^{-2}$, where $m=n^{-1}\sum_{i=1}^nm_i$.

The tuning parameter $h_\mu$ could be selected via the following $\kappa$-fold cross-validation procedure. Let $\kappa$ be a positive integer, e.g., $\kappa=5$, and $\{\mathcal{P}_1,\ldots,\mathcal{P}_\kappa\}$ be a roughly even random partition of the set $\{1,\ldots,n\}$. For a set $\mathcal{H}$ of candidate values for $h_\mu$, we choose one from it such that the following cross-validation error
\begin{equation}\label{eq:mean-cv}
\mathrm{CV}(h)=\sum_{k=1}^\kappa\sum_{i\in \mathcal{P}_k}\sum_{j=1}^{m_i}\{Y_{ij}-\hat{\mu}_{h,-k}(T_{ij})\}^2
\end{equation}
is minimized, where $\hat{\mu}_{h,-k}$ is the estimator in \eqref{eq:ridged-mean} with $h_\mu=h$ and  subjects in $\mathcal{P}_k$ excluded.

\subsection{Covariance Function}

Estimation of the covariance function  $\covarop(s,t)$  for functional snippets is considerably more challenging. As we have pointed out in Section \ref{sec:introduction}, local information in the far off-diagonal region, $| s-t | >  \delta $,  is completely missing. 
To tackle this  challenge, we first observe that the covariance function can be decomposed into two parts, a variance function and a correlation structure, i.e., $\covarop(s,t)=\sigma_X(s)\sigma_X(t)\rho(s,t)$,
where $\sigma_X^2(\cdot)$ is the variance function of $X$, or more precisely, $\sigma_X^2(t)=\expect\{X(t)-\mu(t)\}^2$, and $\rho(\cdot,\cdot)$ is the correlation function. Like the mean function $\mu$, the variance function can be well estimated via local linear smoothing even in the case of functional snippets. The real difficulty stems from the estimation of the correlation structure, which we propose to model parametrically. At  first glance, a parametric model might be restrictive. However, with a nonparametric variance component and a large number of parameters, the model will  often  still  be sufficiently flexible to capture the covariance structure of the data. Indeed, in our simulation studies that are presented in Section \ref{sec:simulation}, we demonstrate that even with a single parameter, the proposed model often  yields good performance when sample size is limited. As an additional flexibility, our parametric model does not require the low-rank assumption 
and hence is able to model truly infinitely-dimensional functional data.  This trade of the low-rank assumption with  the proposed parametric assumption seems worthwhile, especially because we allow the  dimension of the parameters to increase with the sample size.  The increasing dimension of the parameter essentially puts the methodology in the nonparametric paradigm.

To estimate $\sigma_X^2(\cdot)$, we first note that  the PACE  method in \citet{Yao2005a} can still be used to estimate $\covarop(s,t)$ on the  band $\mathcal{T}_{\delta}^2=\{(s,t)\in\tdomain\times\tdomain:\,|s-t|<\delta\}$ that includes the diagonal, although not on the full domain $\tdomain\times \tdomain$. Since $\sigma_X^2(t)=\covarop(t,t)$, the PACE estimate $\tilde{\covarop}$ for $\covarop$ on the diagonal gives rise to an estimate of $\sigma_X^2(t)$. However, this method requires two-dimensional smoothing, which is cumbersome and computationally less efficient. In addition, it has the convergence rate of a two-dimensional smoother, which is suboptimal for a  target $\sigma_X^2(t)$ that is a one-dimensional function.  Here we propose a simpler approach that only requires one-dimensional smoothing, based on the observation that the quantity $\varsigma^2(t)\equiv\expect \{Y(t)-\mu(t)\}^2=\sigma_X^2(t)+\sigma_0^2$ 
can be estimated by local linear smoothing on the observations $\{Y_{ij}-\hat{\mu}(T_{ij})\}^2$.  More specifically, the non-ridged local linear estimate of $\varsigma^2(t)$,  denoted by $\tilde{\varsigma}^2(t)$, is  $\hat{b}_0$ with
\[
(\hat{b}_{0},\hat{b}_{1})=\underset{(b_{0},b_{1})\in\real^{2}}{\arg\min}\sum_{i=1}^{n}w_{i}\sum_{j=1}^{m_{i}}K_{h_{\sigma}}(T_{ij}-t)[\{Y_{ij}-\hat{\mu}(T_{ij})\}^2-b_{0}-b_{1}(T_{ij}-t)]^{2},
\]
where $h_\sigma$ is the bandwidth to be selected by
 a cross-validation procedure similar to \eqref{eq:mean-cv}.
As with the ridged estimate of the mean function in \eqref{eq:ridged-mean}, to circumvent the positive probability of being undefined for $\tilde{\varsigma}^2$, we adopt the ridged version of $\tilde{\varsigma}^2$ as the estimate for $\varsigma^2$, denoted by $\hat{\varsigma}^2$. Then our estimate of $\sigma_X^2(t)$ is  $\hat{\sigma}_X^2(t)=\hat{\varsigma}^2(t)-\hat{\sigma}^2_0$, where $\hat{\sigma}^2_0$ is a new estimate of $\sigma_0^2$,  to be defined in the next section, that has a convergence rate of a one-dimensional smoother.  Because  $\hat{\varsigma}^2(t)$  also has a one-dimensional convergence, the resulting estimate of  $\hat{\sigma}_X^2(t)$ has a one-dimensional convergence rate.

For the correlation function $\rho$, we assume that $\rho$ is indexed by a $d_n$-dimensional vector of parameters, denoted by $\theta\in \real^{d_n}$. Here, the dimension of parameters is allowed to grow with the sample size at a certain rate; see Section \ref{sec:theory} for details. Some popular parametric families for correlation function are listed below. 
\begin{enumerate}
\item Power exponential:
\[
\rho_{\theta}(s,t)=\exp\left\{ -\frac{|s-t|^{\theta_{1}}}{\theta_{2}^{\theta_{1}}}\right\} ,\quad0<\theta_{1}\leq2,\:\theta_{2}>0.
\]
 
\item Rational quadratic (Cauchy):
\[
\rho_{\theta}(s,t)=\left\{ 1+\frac{|s-t|^{2}}{\theta_{2}^{2}}\right\} ^{-\theta_{1}},\quad\theta_{1},\theta_{2}>0.
\]
\item Mat\'ern:
\begin{equation}\label{eq:matern}
\rho_{\theta}(s,t)=\frac{1}{\Gamma(\theta_1)2^{\theta_1-1}}\left(\sqrt{2\theta_1}\frac{|s-t|}{\theta_2}\right)^{\theta_1}B_{\theta_1}\left(\sqrt{2\theta_1}\frac{|s-t|}{\theta_2}\right),\quad\theta_1,\theta_2>0,
\end{equation}
with $B_{\theta}(\cdot)$ being the modified Bessel function of
the second kind of order $\theta$.
\end{enumerate}
Note that if $\rho_1,\ldots,\rho_p$ are correlation functions, then $\sum_{k=1}^p v_k\rho_k$ is also a correlation function if $\sum_{k=1}^p v_k=1$ and $v_k\geq 0$ for all $k$. Therefore, a fairly flexible class of correlation functions can be constructed from several relatively simple classes by  this convex  combination. We point out here that, even when one adopts a stationary correlation function, the resulting covariance  can be non-stationary due to a nonparametric and hence often non-stationary variance component.

Given the estimate $\hat\sigma^2_X(t)$, the parameter $\theta$ can be effectively estimated using the following least squares criterion, i.e., $\hat{\theta}=\underset{\theta}{\arg\min}\,\,\hat{Q}_n(\theta)$ with 
\[
\hat{Q}_n(\theta)=\sum_{i=1}^{n}\frac{1}{m_i(m_i-1)}\sum_{1\leq j\neq l\leq m_{i}}\{\hat{\sigma}_X(T_{ij})\hat{\sigma}_X(T_{il})\rho_{\theta}(T_{ij},T_{il})-C_{ijl}\}^{2},
\]
where $C_{ijl}=\{Y_{ij}-\hat{\mu}(T_{ij})\}\{Y_{il}-\hat{\mu}(T_{il})\}$ is the raw covariance of subject $i$  at two different measurement times, $T_{ij}$  and  $T_{il}$. 

\section{Estimation of Noise Variance}\label{sec:sigma0}

The estimation of $\sigma_0^2$ received relatively little attention in the literature. For sparse functional data, the PACE estimator  $2|\tdomain|^{-1}\int_{\tdomain}\{\hat{\varsigma}(t)-\hat{\covarop}(t,t)\diffop t\}$ proposed in  \citet{Yao2005a} is a popular option. However, the PACE estimator can be negative in some cases. 
\cite{Liu2009} refined this PACE estimator by first fitting the observed data using the PACE estimator and then estimating $\sigma_0^2$  by cross-validated residual sum of squares; see appendix A.1 of \cite{Liu2009} for details. These methods require an estimate of the covariance function, which we do not have here before we obtain an estimate of $\sigma_0^2$. Moreover, the estimate $\hat{\covarop}(t,t)$  in both methods is obtained by two-dimensional local linear smoothing as detailed in \cite{Yao2005a}, which is computationally costly and leads to a slower (two-dimensional) convergence rate of these estimators.
To resolve this conundrum, we propose the following new  estimator that does not require estimation of the covariance function or any other parameters such as the mean function. 

For a bandwidth $h_0>0$, define the quantities
\begin{align*}
A_{0} & =\expect[\{\covarop(T_{1},T_{1})+\mu(T_{1})\mu(T_{1})+\sigma_{0}^{2}\}1_{|T_{1}-T_{2}|<h_0}],\\
A_{1} & =\expect[\{\covarop(T_{1},T_{1})+\mu(T_{1})\mu(T_{1})\}1_{|T_{1}-T_{2}|<h_0}],
\end{align*}
and $$B=\expect1_{|T_{1}-T_{2}|<h_0},$$ where  $T_{1}$ and $T_{2}$ denote two design points from the same generic subject. From the above definition, we immediately see that $A_0=A_{1}+B\sigma_{0}^{2}$. Also, these quantities seem easy to estimate. For example, $A_0$ and $B$ can be straightforwardly estimated respectively by 
\begin{equation}\label{eq:A0-est}
\hat{A}_{0}=\frac{1}{n}\sum_{i=1}^{n}\frac{1}{m_{i}(m_{i}-1)}\sum_{j\neq l}Y_{ij}^{2}1_{|T_{ij}-T_{il}|<h_0}
\end{equation}
and
\begin{equation}\label{eq:B-est}
\hat{B}=\frac{1}{n}\sum_{i=1}^{n}\frac{1}{m_{i}(m_{i}-1)}\sum_{j\neq l}1_{|T_{ij}-T_{il}|<h_0}.
\end{equation}
This motivates us to  estimate $\sigma_0^2$ via estimation of $A_0$, $A_1$ and $B$. 

It remains to estimate $A_1$, which cannot be estimated using information along the diagonal only, due to the presence of random noise. Instead, we shall explore the smoothness of the covariance function and observe that if $T_1$ is close to $T_2$, say $|T_1-T_2|<h_0$, then $\covarop(T_1,T_1)\approx \covarop(T_1,T_2)$ and 
$$A_1\approx A_{2} =\expect[\{\covarop(T_{1},T_{2})+\mu(T_{1})\mu(T_{2})\}1_{|T_{1}-T_{2}|<h_0}].$$
Indeed, we show in Lemma \ref{lem:A0-A1-A2-B} that $A_{1}=A_{2}+O(h_0^{3})$. Therefore, it is sensible to use $A_2$ as a surrogate of $A_1$. The former can be effectively estimated by 
\begin{equation}\label{eq:A1-est}
\hat{A}_{2}=\frac{1}{n}\sum_{i=1}^{n}\frac{1}{m_{i}(m_{i}-1)}\sum_{j\neq l}Y_{ij}Y_{il}1_{|T_{ij}-T_{il}|<h_0},
\end{equation}
and we set $\hat{A}_1=\hat{A}_2$. 
Finally, the estimate of $\sigma_{0}^{2}$ is given by
\begin{equation}\label{eq:simga0}
\hat{\sigma}_{0}^{2}=(\hat{A}_{0}-\hat{A}_{1})/\hat{B}.
\end{equation}

To choose $h_0$, motivated by the convergence rate stated in Theorem \ref{thm:sigma0} of the next section, we suggest the following empirical rule,  $h_0=0.29\hat{\delta}\|\hat\varsigma\|_2(nm^2)^{-1/5}$,  for sparse functional snippets, where $\hat{\delta}=\max_{1\leq i\leq n}\max_{1\leq j,l \leq m_i}|T_{ij}-T_{il}|$ acts as an estimate for $\delta$, $m=n^{-1}\sum_{i=1}^n m_i$ represents the average number of measurements per curve, $\hat\varsigma^2(t)$ is the estimate of $\varsigma^2(t)=\sigma_X^2(t)+\sigma_{0}^2$ defined in Section \ref{sec:method}, and $\|\hat\varsigma\|_2^2=\int \hat\varsigma^2(t)\diffop t$ represents the overall variability of the data. The coefficient $0.29$ is determined by a method described in the appendix. If this rule yields a value of $h_0$ that makes the neighborhood $\mathcal{N}(h_0)=\{(T_{ij},T_{il}):|T_{ij}-T_{il}|<h_0,i=1,\ldots,n,1\leq j\neq l\leq m_i\}$ empty or contain too few points, then we recommend to choose the minimal value of $h_0$ such that $\mathcal{N}(h_0)$ contains at least $10^{-1} \sum_{i=1}^n m_i(m_i-1)$ points. In this way, we ensure that a substantial fraction of the observed data are used for estimation of the variance $\sigma_0^2$. This rule is found to be very effective in practice; see Section \ref{sec:simulation} for its numerical performance.

Compared to \citet{Yao2005a} and \cite{Liu2009}, the proposed estimate (\ref{eq:simga0}) is simple and easy to compute. Indeed, it can be computed much faster since it does not require the costly computation of $\hat{\covarop}$. More importantly, the ingredients $\hat{A}_0$, $\hat{A}_1=\hat{A}_2$ and $\hat{B}$ for our estimator are obtained by one-dimensional smoothing, with the term $1_{|T_{ij}-T_{il}|<h_0}$ in \eqref{eq:A0-est}--\eqref{eq:A1-est} acting as a local constant smoother. Consequently, as we show in Section \ref{sec:theory}, our estimator enjoys an asymptotic convergence rate that is faster than the one from a two-dimensional local linear smoother. In addition, the proposed estimate is always nonnegative, in contrast to the one in \citet{Yao2005a}. This is seen by the following derivation:
\begin{align}
\hat{A}_{1} & =\frac{1}{n}\sum_{i=1}^{n}\frac{1}{m_{i}(m_{i}-1)}\sum_{j\neq l}Y_{ij}Y_{il}1_{|T_{ij}-T_{il}|<h_0}\leq\frac{1}{n}\sum_{i=1}^{n}\frac{1}{m_{i}(m_{i}-1)}\sum_{j\neq l}\frac{Y_{ij}^{2}+Y_{il}^{2}}{2}1_{|T_{ij}-T_{il}|<h_0} \nonumber\\
 & =\frac{1}{n}\sum_{i=1}^{n}\frac{1}{m_{i}(m_{i}-1)}\sum_{j\neq l}Y_{ij}^{2}1_{|T_{ij}-T_{il}|<h_0} =\hat{A}_{0}. \label{eq:sigma0-positive}
\end{align}

\noindent \textbf{Remark:} The above discussion assumes that the noise is homoscedastic, i.e., its variance is identical for all $t$. As an extension, it is possible to modify the above procedure to account for heteroscedastic noise, as follows. With intuition and rationale similar to the homoscedastic case, we define 
\begin{align*}
\hat{A}_{0}(t) & =\frac{1}{n}\sum_{i=1}^{n}\frac{1}{m_{i}(m_{i}-1)}\sum_{j\neq l}Y_{ij}^{2}1_{|T_{ij}-t|<h_0}1_{|T_{il}-t|<h_0},\\
\hat{A}_{1}(t) & =\frac{1}{n}\sum_{i=1}^{n}\frac{1}{m_{i}(m_{i}-1)}\sum_{j\neq l}Y_{ij}Y_{il}1_{|T_{ij}-t|<h_0}1_{|T_{il}-t|<h_0},\\
\hat{B}(t) & =\frac{1}{n}\sum_{i=1}^{n}\frac{1}{m_{i}(m_{i}-1)}\sum_{j\neq l}1_{|T_{ij}-t|<h_0}1_{|T_{il}-t|<h_0},
\end{align*}
and let 
\begin{equation*}
\hat{\sigma}_{0}^{2}(t)=\{\hat{A}_{0}(t)-\hat{A}_{1}(t)\}/\hat{B}(t)
\end{equation*}
be the estimate of $\sigma_0^2(t)$ which is the variance of the noise at $t\in\tdomain$. Like the derivation in \eqref{eq:sigma0-positive}, one can also show that this estimator is nonnegative.

\section{Theoretical Properties}\label{sec:theory}
For clarity of exposition, we  assume throughout this section that all the $m_i$ have the same rate $m$, i.e., $m_i = m$,   where the sampling rate $m$  may tend to infinity.   We emphasize that parallel asymptotic results can be derived without this assumption by replacing $m$ with $ \frac{1}{n}\sum_{i=1}^n m_i$. Note that the theory to be presented below applies to both the case that $m$ is bounded by a constant, i.e., $m\leq m_0$ for some $m_0<\infty$, and the case that $m$ diverges to $\infty$ as $n\rightarrow\infty$. 

We assume that the reference time $O_i$ is identically and independently distributed (i.i.d.) sampled from a density $f_O$, and  $T_{i1},\ldots,T_{im_{i}}$ are i.i.d., conditional on $O_i$.  The i.i.d. assumptions can be relaxed to accommodate heterogeneous distributions and weak dependence, at the cost of much more complicated analysis and heavy technicalities. As such relaxation does not provide further insight into our problem, we decide not to pursue it in the paper.  The following conditions about $O_i$ and other quantities are needed for our theoretical development.
\begin{enumerate}
[leftmargin=*,labelsep=7mm,label= (A\arabic*)]
\setcounter{enumi}{0}%
\item\label{cond:A1}The density $f_O$ of each $O_i$ satisfies $f_{O}(u)>0$ for all $u\in[\delta/2,1-\delta/2]$, and the conditional density $f_{T|O}$ of $T_{ij}$ given $O_i$ satisfies $f_{T|O}(t|u)=f_0(t-u+\delta/2)>0$ for a fixed function $f_0$ and 
for all $u\in[\delta/2,1-\delta/2]$ and $t\in[u-\delta/2,u+\delta/2]$.
Also, the derivative $\frac{\diffop}{\diffop t}f_0(t)$ is Lipschitz continuous on $[0,\delta]$. 
\item\label{cond:A2}The second derivatives of $\mu$ and $\covarop$ 
are continuous and hence bounded on $\tdomain$ and $\tdomain\times \tdomain$, respectively.
\item\label{cond:A3}$\expect\|X\|^{4}<\infty$ and $\expect\varepsilon^{4}<\infty$.
\end{enumerate}
In the above, the condition \ref{cond:A1} characterizes the design points for functional snippets and can be relaxed, while the regularity conditions \ref{cond:A2} and \ref{cond:A3} are  common in the literature, e.g., in \citet{Zhang2016}. According to \cite{Scheuerer2010}, \ref{cond:A2} also implies that the sample paths of $X$ are continuously differentiable and hence Lipschitz continuous almost surely. Let $L_X$ be the best Lipschitz constant of $X$, i.e., $L_X=\inf\{C\in\real:\,|X(s)-X(t)|\leq C|s-t|\text{ for all }s,t\in\tdomain\}$. We will see shortly that a moment condition on $L_X$ allows us to derive a rather sharp bound for the convergence rate of $\hat{\sigma}_0^2$.
For the bandwidth $h_0$, we require the following condition:
\begin{enumerate}
[leftmargin=*,labelsep=7mm,label= (H\arabic*)]
\setcounter{enumi}{0}%
\item\label{cond:H1}$h_0\rightarrow0$ and $nm^2h_0\rightarrow\infty$.
\end{enumerate}

The following result gives the asymptotic rate of the estimator $\hat{\sigma}_0^2$. The proof is straightforward once we have Lemma \ref{lem:A0-A1-A2-B}, which is given in the appendix.
\begin{theorem}\label{thm:sigma0}Assume the conditions \ref{cond:A1}--\ref{cond:A3} and \ref{cond:H1} hold.
\begin{enumerate}
[leftmargin=*,labelsep=2mm,label= \textup{(\alph*)}]
\item\label{thm:sigma0:a} $(\hat{\sigma}_{0}^{2}-\sigma_{0}^{2})^{2}=\Op(h_0^{4}+n^{-1}+n^{-1}m^{-2}h_0^{-1})$. With the optimal choice $h_0\asympeq (nm^2)^{-1/5}$, $(\hat{\sigma}_{0}^{2}-\sigma_{0}^{2})^{2}=\Op(n^{-4/5}m^{-8/5}+n^{-1})$.
\item\label{thm:sigma0:b} If in addition $\expect L_X^4<\infty$, then $(\hat{\sigma}_{0}^{2}-\sigma_{0}^{2})^{2}=\Op(h_0^{4}+n^{-1}m^{-1}+n^{-1}m^{-2}h_0^{-1})$. With the optimal choice $h_0\asympeq (nm^2)^{-1/5}$, $(\hat{\sigma}_{0}^{2}-\sigma_{0}^{2})^{2}=\Op(n^{-4/5}m^{-8/5}+n^{-1}m^{-1})$.
\end{enumerate}
\end{theorem}
If we define  $\hat\sigma_0^2=(\hat A_0-\hat A_1)/(\hat B+\Delta)$ with $\Delta=(nm)^{-2}h_0$, the ridged version of \eqref{eq:simga0}, then in the above theorem, $(\hat{\sigma}_{0}^{2}-\sigma_{0}^{2})^{2}$ can be replaced with $\expect (\hat{\sigma}_{0}^{2}-\sigma_{0}^{2})^{2}$ and $\Op(\cdot)$ can be replaced with $O(\cdot)$, respectively.  For comparison, under conditions stronger than \ref{cond:A1}--\ref{cond:A3}, the rate derived in \cite{Yao2005a} for the PACE estimator is at best  $(\hat{\sigma}_{0}^{2}-\sigma_{0}^{2})^{2}=\Op(n^{-1/2})$.  This rate  was improved by \cite{Paul2011} to $\expect(\hat{\sigma}_{0}^{2}-\sigma_{0}^{2})^{2}=O(n^{-1}+n^{-4/5}m^{-4/5}+n^{-2/3}m^{-4/3})$. Our estimator clearly enjoys a faster convergence rate, in addition to its computational efficiency. The rate in part \ref{thm:sigma0:b} of Theorem \ref{thm:sigma0} has little room for improvement, since when $n$ is finite but $m\rightarrow\infty$, the rate is optimal, i.e., $\expect(\hat{\sigma}_0^2-\sigma_0^2)^2=O(m^{-1})$.   When $m$ is finite but $n\rightarrow\infty$ in the sparse design,  we obtain $\expect(\hat{\sigma}_0^2-\sigma_0^2)^2=O(n^{-4/5})$, in contrast to  the rate $O_P(n^{-2/3})$ for the PACE estimator according to \cite{Paul2011}.

To study the properties of $\hat{\mu}(t)$ and $\hat{\sigma}^2(t)$, we shall assume
\begin{enumerate}
[leftmargin=*,labelsep=7mm,label= (B\arabic*)]
\setcounter{enumi}{0}%
\item\label{cond:B1}the kernel $K$ is a symmetric and Lipschitz continuous
density function supported on $[-1,1]$.
\end{enumerate}
Also, the bandwidth $h_\mu$ and $h_\sigma$ are assumed to meet the following conditions.
\begin{enumerate}
[leftmargin=*,labelsep=7mm,label= (H\arabic*)]
\setcounter{enumi}{1}%
\item\label{cond:H2}$h_{\mu}\rightarrow0$ and $nmh_{\mu}\rightarrow\infty$.
\item\label{cond:H3} $h_{\sigma}\rightarrow0$ and $nmh_{\sigma}\rightarrow\infty$.
\end{enumerate}
The choice of these bandwidths depends on the interplay of the sampling rate $m$ and sample size $n$. The optimal choice is given in the following condition.
\begin{enumerate}
[leftmargin=*,labelsep=7mm,label= (H\arabic*)]
\setcounter{enumi}{3}%
\item\label{cond:H4}If $m\asymplt n^{1/4}$,  then $h_\mu\asympeq h_\sigma\asympeq (nm)^{-1/5}$, where the notation $a_n\asymplt b_n$ means $\lim_{n\rightarrow\infty}a_n/b_n<\infty$. Otherwise, $\max\{h_\mu,h_\sigma\}\asympeq n^{-1/4}$. Also, $h_0\asympeq (nm^2)^{-1/5}$.
\end{enumerate}
The asymptotic convergence rates for $\hat{\mu}$ and $\hat{\sigma}_X^2$ are given in the following theorem, whose proof can be obtained by adapting the proof of Proposition 1 in  \citet{Lin2020+} and hence is omitted. It shows that both $\hat{\mu}$ and $\hat{\sigma}_X^2$ have the same rate, which is hardly surprising since they are both obtained by a one-dimensional local linear smoothing technique. Note that our results generalize those in \cite{Fan2007} and \cite{Fan2008} by taking the measurement errors and the order of the sampling rate $m$ into account in the theoretical analysis. In addition, our $\ltwo$ convergence rates of these estimators complement the asymptotic normality results in  \cite{Fan2007} and \cite{Fan2008}.
\begin{theorem}\label{thm:mu-sigma}
Suppose the conditions \ref{cond:A1}--\ref{cond:A3} hold.
\begin{enumerate}
[leftmargin=*,labelsep=2mm,label= \textup{(\alph*)}]
\item With additional conditions \ref{cond:B1} and \ref{cond:H2}, $\expect\|\hat{\mu}-\mu\|^{2}=O(h_{\mu}^{4}+n^{-1}+n^{-1}m^{-1}h_{\mu}^{-1})$. With the choice of bandwidth $h_\mu$ in \ref{cond:H4}, $\expect\|\hat{\mu}-\mu\|^{2}=O\left((nm)^{-4/5}+n^{-1}\right)$.
\item With additional conditions \ref{cond:B1} and \ref{cond:H1}--\ref{cond:H3}, $\expect\|\hat{\sigma}_X^{2}-\sigma_X^{2}\|^{2}=O(h_{\sigma}^{4}+h_\mu^4+h_0^{4}+n^{-1}+n^{-1}m^{-1}h_{\sigma}^{-1}+n^{-1}m^{-1}h_{\mu}^{-1}+n^{-1}m^{-2}h_0^{-1})$. With the choice of bandwidth in \ref{cond:H4}, $\expect\|\hat{\sigma}_X^{2}-\sigma_X^{2}\|^{2}=O\left((nm)^{-4/5}+n^{-1}\right)$.
\end{enumerate}
\end{theorem}

To derive the asymptotic properties of $\hat{\covarop}(s,t)=\hat{\sigma}_X(s)\rho_{\hat{\theta}}(s,t)\hat{\sigma}_X(t)$, we need the convergence rate of $\hat{\theta}$. Define
\begin{align*}
Q(\theta) & =\expect\{\sigma_X(T_{11})\sigma_X(T_{12})\rho_{\theta}(T_{11},T_{12})-[Y_{11}-\mu(T_{11})][Y_{12}-\mu(T_{12})]\}^{2},
\end{align*}
and assume the following conditions.
\begin{enumerate}
[leftmargin=*,labelsep=7mm,label= (B\arabic*)]
\setcounter{enumi}{1}%
\item\label{cond:B2} $\rho_{\theta}(s,t)$ is twice continuously differentiable
with respect to $s$ and $t$. Furthermore, the first three derivatives of $\rho_\theta(s,t)$ with respect to $\theta$ are uniformly bounded for all $\theta, s ,t,d_n$.
\item\label{cond:B3} $\lambda_{\min}\left(\frac{\partial^{2}Q}{\partial\theta^{2}}\mid_{\theta=\theta_{0}}\right)>c_0d_n^{-\tau}$ for some $c_0>0$ and $\tau\geq0$,
where $\theta_0$ denotes the true value of $\theta$, and $\lambda_{\min}(\cdot)$ denotes the smallest eigenvalue of
a matrix.
\item\label{cond:B4}$\expect\sup_t\|X(t)\|^{4+\epsilon_0}<\infty$ for some $\epsilon_0>0$ and $\expect\varepsilon^{4}<\infty$.
\end{enumerate}
Note that in the condition \ref{cond:B3}, we allow the smallest eigenvalue of the Hessian $\frac{\partial^{2}Q}{\partial\theta^{2}}$ to decay with $d_n$. This, departing from the assumption in \cite{Fan2008} of fixed dimension on the parameter $\theta$, enables us to incorporate the case that $\rho_\theta$ is constructed from the aforementioned convex combination of a \emph{diverging} number of correlation functions, e.g., $\rho_\theta(s,t)=d^{-1}_n\sum_{j=1}^{d_n}\rho_{\theta_j}(s,t)$, where $\tau=1$ if all components $\rho_{\theta_j}$ satisfy \ref{cond:B2} uniformly. The condition \ref{cond:B4}, although it  is slightly stronger than \ref{cond:A3}, is often required in  functional data analysis, e.g., in  \cite{Li2010} and \cite{Zhang2016} for the derivation of uniform convergence rates for $\hat{\mu}$. Such uniform rates are required to bound $\partial\hat{Q}_n/\partial\theta$ sharply in our development, which is critical to establish the following rate for $\hat{\theta}$.
\begin{proposition}
\label{thm:theta}Suppose  the conditions
\ref{cond:A1}--\ref{cond:A2} and \ref{cond:B1}--\ref{cond:B4} hold. If $d_n=o(n^{1/(4+4\tau)})$, then   with the choice of bandwidth in \ref{cond:H4}, $\|\hat{\theta}-\theta_{0}\|^{2}=\Op(d_n^{2\tau+1}/n)$.
\end{proposition}
The above result suggests that the estimation quality of $\hat{\theta}$ depends on the dimension of parameters, sample size and singularity of the Hessian matrix at $\theta=\theta_0$, measured by the constant $\tau$ in condition \ref{cond:B3}. In practice, a few  parameters are often sufficient for an adequate fit. In such cases, the dimension $d_n$ might not grow with sample size, i.e., $d_n=O(1)$, and we obtain a parametric rate for $\hat{\theta}$. Now we are ready to state our main theorem that establishes the convergence rate for $\hat{\covarop}$ in the Hilbert-Schmidt norm $\|\cdot\|_{HS}$, which follows immediately from the above results.
\begin{theorem}\label{thm:cov}
Under the same conditions of Proposition \ref{thm:theta}, we have
$\|\hat{\covarop}-\covarop\|_{HS}^{2}=\Op(h_{\sigma}^{4}+h_\mu^4+h_0^{4}+n^{-1}+n^{-1}m^{-1}h_{\sigma}^{-1}+n^{-1}m^{-1}h_{\mu}^{-1}+n^{-1}m^{-2}h_0^{-1}+d_n^{2\tau+1}n^{-1}).$
With the choice of bandwidth in \ref{cond:H4}, $\|\hat{\covarop}-\covarop\|_{HS}^{2}=\Op\left((nm)^{-4/5}+d_n^{2\tau+1}n^{-1}\right)$.
\end{theorem}

In practice, a fully nonparametric approach like local regression to estimating the correlation structure is inefficient, in particular when data are snippets. On the other hand, a parametric method with a fixed number of parameters might be restrictive when the sample size is large. One way to overcome such a dilemma is to allow the family of parametric models to grow with the sample size. As a working assumption, one might consider that the correlation function $\rho$ falls into $\mathcal{F}_n$, a $d_n$-dimensional family of models for correlation functions, when the sample size is $n$. Here, the dimension typically grows with the sample size. For example, one might consider a $d_n$-Fourier basis family:
\begin{equation}\label{eq:basis-expansion}
\kappa_{\theta}(s,t)=\frac{1}{\psi(s)\psi(t)}\sum_{j=1}^{d_n} \theta_j\phi_j(s)\phi_j(t),\quad\theta_1,\ldots,\theta_{d_n} \geq 0 \text{ and } \sum_{j=1}^{d_n}\theta_j=1,
\end{equation}
where $\psi(t)=\left(\sum_{j=1}^{d_n}\theta_j\phi_j^2(t)\right)^{1/2}$ and $\phi_1,\ldots$ are fixed  orthonormal Fourier basis functions defined on $\tdomain$. The theoretical result in Theorem \ref{thm:cov} applies to this setting by explicitly accounting for the impact of the dimension $d_n$ on the convergence rate. 

\section{Simulation Studies}\label{sec:simulation}
To evaluate the numerical performance of the proposed estimators, 
we generated $X(\cdot)$ from a Gaussian process. 
Three different covariance functions were considered, namely,
\begin{itemize}
	\item[I.] $\covarop(s,t)=\sigma_X(s)\rho_\theta(s,t)\sigma_X(t)$ with the variance function $\sigma^2_X(t)=\sqrt{t}e^{-(t-0.1)^2/10}+1$ and the Mat\'ern correlation function $\rho_{\theta=(0.5,1)}$, 
	\item[II.] $\covarop(s,t)=\sum_{k=1}^{50} 2k^{-\lambda}\phi_k(s)\phi_k(t)$ with $\lambda=2$ and Fourier basis functions $\phi_k(t)=\sqrt{2}\sin(2k\pi t)$, and
	\item[III.] $\covarop(s,t)=\sum_{1\leq j,k\leq 5} c_{jk}\phi_j(s)\phi_k(t)$ with $c_{jk}=e^{-|j-k|}/5$.
\end{itemize}
 Two different sample sizes $n=50$ and $n=200$ were considered to illustrate the behavior of the estimators under a small sample size and a relatively large sample size. We set the domain $\mathcal{T}=[0,1]$ and $\delta=0.25$.  
 
 To evaluate the impact of the mean function, we also considered two different mean functions,  $\mu_1(t)=2t^2\cos(2\pi t)$ and $\mu_2(t)=e^t/2$. We found that the results are not sensitive to the mean function, and thus focus only on the case $\mu_1(t)$ in this section; the results for the case $\mu(t)=e^t/2$ are provided in Supplementary Material. In addition, to evaluate the impact of the design, we considered two design schemes. In the first scheme, that is referred to as the sparse design, each curve was sparsely sampled at 4 points on average to mimic the scenario of the data application in Section \ref{sec:application}. In the second scheme, that is referred to as the dense design, each snippet was recorded in a dense ($m_1=\cdots=m_n=26$) and regular grid of an individual specific subinterval of length $\delta$. As the focus of the paper is on sparse snippets, we report the results for the sparse design below. The results for dense snippets are reported in Supplementary Material.

To assess the performance of the estimators for the noise variance $\sigma_0^2$, we considered different noise levels  $\sigma_0^2=0,0.1,0.25,0.5$, varying from no noise to large noise. For example,  when $\sigma_0^2=0.5$, the signal-to-noise ratio $\expect\|X-\mu\|^2/\var(\varepsilon)$ is about 2.    The performance of $\hat{\sigma}_0^2$  is assessed by the root mean squared error (RMSE), defined by $$\mathrm{RMSE}=\sqrt{\frac{1}{N}\sum_{i=1}^N|\hat{\sigma}_0^2-\sigma_0^2|^2},$$ where $N$ is the number of independent simulation replicates, which we set to 100. For the purpose of comparison, we also computed the PACE estimate of \citet{Yao2005a} and the estimate proposed by \citet{Liu2009}, denoted by LM, using the 
{fdapace} R package \citep{fdapace} that is available in the comprehensive R archive network (CRAN). The bandwidth $h_\mu$ and $h_\sigma$, as well as those in \cite{Yao2005a} and \cite{Liu2009}, were selected by five-fold cross-validation. The tuning parameter $h_0$ was selected by the empirical rule $h_0=0.29\hat\delta\|\hat\varsigma\|_2(nm^2)^{-1/5}$ that is described in Section \ref{sec:sigma0}. The simulation results are summarized in Table \ref{tab:sig-sparse-1} for the sparse design with mean function $\mu_1$, as well as Tables S.1--S.3 for the dense design and mean function $\mu_2$ in Supplementary Material, where SNPT denotes our  method proposed in Section \ref{sec:sigma0}. We observe that in almost all cases, SNPT performs significantly better than the other two methods. The results also demonstrate the effectiveness of the proposed empirical selection rule for the tuning parameter $h_0$.

To evaluate the performance of the estimators for the covariance structure, we considered two levels of signal-to-noise ratio (SNR), namely, $\mathrm{SNR}=2$ and $\mathrm{SNR}=4$. The performance of estimators for the variance function and the covariance function 
is  evaluated by the root mean integrated squared error (RMISE), defined by $$\mathrm{RMISE}=\sqrt{\frac{1}{N}\sum_{i=1}^N \int_{\tdomain}|\hat{\sigma}^2_X(t)-\sigma^2_X(t)|^2\diffop t}$$
for the variance function 
and 
$$\mathrm{RMISE}=\sqrt{\frac{1}{N}\sum_{i=1}^N \int_{\tdomain}\int_{\tdomain}|\hat{\covarop}(s,t)-\covarop(s,t)|^2\diffop s\diffop t}$$
for the covariance function. 
We compared four methods. The first two, denoted by SNPTM and SNPTF, are our semi-parametric approach with the correlation given in \eqref{eq:matern} and  \eqref{eq:basis-expansion}, respectively. For the SNPTF method, the dimension $d_n$ of \eqref{eq:basis-expansion} is selected via five-fold cross-validation. It is noted that SNPTM and SNPTF yield the same estimates of the variance function but different estimates of the correlation structure.  The third one, denoted by PFBE (penalized Fourier basis expansion), is the method proposed by \cite{Lin2019a}, and the last one, denoted by PACE, is the approach invented by \cite{Yao2005a}.

For the estimation of the variance function $\sigma_X^2(t)$, the results are summarized in Table \ref{tab:var-sparse-1} for the sparse design and mean function $\mu_1$, and also in Tables S.4--S.6 in Supplementary Material for the dense design and mean function $\mu_2$. In these tables, the results of SNPTF are not reported since they are the same as the results of SNPTM.  We observe that, in all cases, SNPTM and PFBE substantially outperform PACE. For the dense design, the methods SNPTM and PFBE yield comparable results. The SNPTM method performs better than PFBE when $n=200$ in most cases, except in the setting III which favors the PFBE method. This suggests that the SNPTM method, which adopts the local linear smoothing strategy combined with our estimator for the variance of the noise, generally converges faster as the sample size grows.

For the estimation of the covariance function $\covarop$, we summarize the results in Table \ref{tab:cov-sparse-1} for the sparse design and mean function $\mu_1$, and in Tables S.7--S.9 in Supplementary Material for the dense design and mean function $\mu_2$.   As expected, in all cases, SNPTM, SNPTF and PFBE substantially outperform PACE, since PACE is not designed to process functional snippets. Among the estimators SNPTM, SNPTF and PFBE, in the setting I, SNPTM outperforms the others since in this case the model is correctly specified for SNPTM, in the setting II, SNPTF is the best since the model is correctly specified for SNPTF, and in the setting III, PFBE has a favorable performance. Although there is no universally best estimator, overall these three estimators have comparable performance. To select a method in practice, one can first produce a scatter plot of the raw covariance function. If the function appears to decay monotonically as the point $(s,t)$ moves away from the diagonal, then SNPT with a monotonic decaying correlation such as SNPTM is recommended. Otherwise, SNPT with a general correlation structure such as SNPTF or the PFBE approach might be adopted.

\begin{table}[h!]
	\renewcommand{\arraystretch}{0.8}
	\caption{RMSE and their standard errors for $\hat{\sigma}_{0}^{2}$ under the sparse design and $\mu_1$ }
	\begin{center}
	\begin{tabular}{|c|c|c|c|c|c|}
		\cline{4-6}  \cline{5-6} \cline{6-6}
		\multicolumn{3}{c|}{} & \multicolumn{3}{c|}{method}\tabularnewline
		\hline
		Cov & $n$ & $\sigma_0^2$ & SNPT & PACE & LM\tabularnewline
		\hline
		\hline
		\multirow{8}{*}{I} & \multirow{4}{*}{50} & 0  & 0.012 (0.009) & 0.144 (0.166) & 0.129 (0.203) \tabularnewline
		\cline{3-6}
		& & 0.1  & 0.029 (0.038) & 0.129 (0.146) & 0.186 (0.197) \tabularnewline
		\cline{3-6}
		& & 0.25  & 0.050 (0.056) & 0.147 (0.185) & 0.117 (0.125) \tabularnewline
		\cline{3-6}
		& & 0.5  & 0.100 (0.135) & 0.181 (0.195) & 0.157 (0.131) \tabularnewline
		\cline{2-6}
		& \multirow{4}{*}{200} & 0  & 0.009 (0.005) & 0.080 (0.103) & 0.073 (0.077) \tabularnewline
		\cline{3-6}
		& & 0.1  & 0.017 (0.019) & 0.091 (0.098) & 0.144 (0.150) \tabularnewline
		\cline{3-6}
		& & 0.25  & 0.032 (0.038) & 0.086 (0.097) & 0.093 (0.127) \tabularnewline
		\cline{3-6}
		& & 0.5  & 0.049 (0.064) & 0.098 (0.118) & 0.165 (0.106) \tabularnewline
		\hline
		\multirow{8}{*}{II} & \multirow{4}{*}{50} & 0  & 0.036 (0.030) & 0.252 (0.245) & 0.219 (0.255) \tabularnewline
		\cline{3-6}
		& & 0.1  & 0.047 (0.052) & 0.254 (0.285) & 0.237 (0.255) \tabularnewline
		\cline{3-6}
		& & 0.25  & 0.087 (0.133) & 0.241 (0.244) & 0.159 (0.151) \tabularnewline
		\cline{3-6}
		& & 0.5  & 0.128 (0.202) & 0.238 (0.260) & 0.126 (0.134) \tabularnewline
		\cline{2-6}
		& \multirow{4}{*}{200} & 0  & 0.024 (0.015) & 0.177 (0.172) & 0.192 (0.200) \tabularnewline
		\cline{3-6}
		& & 0.1  & 0.027 (0.027) & 0.185 (0.179) & 0.176 (0.174) \tabularnewline
		\cline{3-6}
		& & 0.25  & 0.042 (0.050) & 0.177 (0.177) & 0.097 (0.097) \tabularnewline
		\cline{3-6}
		& & 0.5  & 0.071 (0.084) & 0.174 (0.182) & 0.124 (0.089) \tabularnewline
		\hline
		\multirow{8}{*}{III} & \multirow{4}{*}{50} & 0  & 0.004 (0.004) & 0.099 (0.103) & 0.028 (0.064) \tabularnewline
		\cline{3-6}
		& & 0.1  & 0.024 (0.029) & 0.102 (0.106) & 0.099 (0.127) \tabularnewline
		\cline{3-6}
		& & 0.25  & 0.049 (0.063) & 0.093 (0.109) & 0.077 (0.080) \tabularnewline
		\cline{3-6}
		& & 0.5  & 0.094 (0.130) & 0.113 (0.146) & 0.172 (0.128) \tabularnewline
		\cline{2-6}
		& \multirow{4}{*}{200} & 0  & 0.002 (0.002) & 0.065 (0.077) & 0.009 (0.023) \tabularnewline
		\cline{3-6}
		& & 0.1  & 0.010 (0.012) & 0.066 (0.067) & 0.049 (0.075) \tabularnewline
		\cline{3-6}
		& & 0.25  & 0.027 (0.033) & 0.068 (0.071) & 0.069 (0.067) \tabularnewline
		\cline{3-6}
		& & 0.5  & 0.059 (0.071) & 0.067 (0.073) & 0.163 (0.091) \tabularnewline
		\hline
	\end{tabular}
\end{center}
	\label{tab:sig-sparse-1}
\end{table}

\begin{table}
	\renewcommand{\arraystretch}{0.8}
	\caption{RMISE and their standard errors for $\hat{\sigma}^{2}_X(t)$ under the sparse design and $\mu_1$}
	\begin{center}
	\begin{tabular}{|c|c|c|c|c|c|}
		\cline{4-6}  \cline{5-6} \cline{6-6}
		\multicolumn{3}{c|}{} & \multicolumn{3}{c|}{method}\tabularnewline
		\hline
		Cov & SNR & $n$ & SNPTM & PFBE & PACE\tabularnewline
		\hline
		\hline
		\multirow{4}{*}{I} & \multirow{2}{*}{2} & 50  & 0.535 (0.218) & 0.518 (0.211) & 2.133 (1.536) \tabularnewline
		\cline{3-6}
		& & 200  & 0.339 (0.130) & 0.330 (0.118) & 1.344 (1.126) \tabularnewline
		\cline{2-6}
		& \multirow{2}{*}{4} & 50  & 0.531 (0.199) & 0.517 (0.229) & 1.845 (1.461) \tabularnewline
		\cline{3-6}
		& & 200  & 0.313 (0.136) & 0.334 (0.127) & 1.151 (0.952) \tabularnewline
		\hline
		\multirow{4}{*}{II} & \multirow{2}{*}{2} & 50  & 0.775 (0.396) & 0.743 (0.214) & 2.602 (1.747) \tabularnewline
		\cline{3-6}
		& & 200  & 0.509 (0.163) & 0.530 (0.141) & 1.699 (1.045) \tabularnewline
		\cline{2-6}
		& \multirow{2}{*}{4} & 50  & 0.768 (0.303) & 0.734 (0.351) & 2.510 (1.578) \tabularnewline
		\cline{3-6}
		& & 200  & 0.471 (0.162) & 0.507 (0.149) & 1.515 (1.056) \tabularnewline
		\hline
		\multirow{4}{*}{III} & \multirow{2}{*}{2} & 50  & 0.633 (0.201) & 0.592 (0.136) & 1.478 (1.052) \tabularnewline
		\cline{3-6}
		& & 200  & 0.376 (0.133) & 0.392 (0.107) & 1.178 (0.700) \tabularnewline
		\cline{2-6}
		& \multirow{2}{*}{4} & 50  & 0.592 (0.208) & 0.586 (0.158) & 1.428 (1.166) \tabularnewline
		\cline{3-6}
		& & 200  & 0.350 (0.139) & 0.385 (0.114) & 0.923 (0.451) \tabularnewline
		\hline
	\end{tabular}
\end{center}
	\label{tab:var-sparse-1}
\end{table}

\begin{table}
	\renewcommand{\arraystretch}{0.8}
	\caption{RMISE and their standard errors for $\hat\covarop$ under the sparse design and $\mu_1$}
	\begin{center}
	\begin{tabular}{|c|c|c|c|c|c|c|}
		\cline{4-7}  \cline{5-7} \cline{6-7} \cline{7-7} 
		\multicolumn{3}{c|}{} & \multicolumn{4}{c|}{method}\tabularnewline
		\hline
		Cov & SNR & $n$ & SNPTM & SNPTF & PFBE & PACE\tabularnewline
		\hline
		\hline
		\multirow{4}{*}{I} & \multirow{2}{*}{2} & 50  & 0.339 (0.101) & 0.441 (0.158) & 0.399 (0.156) & 1.470 (0.808) \tabularnewline
		\cline{3-7}
		& & 200  & 0.235 (0.092) & 0.359 (0.089) & 0.295 (0.101) & 1.044 (0.625) \tabularnewline
		\cline{2-7}
		& \multirow{2}{*}{4} & 50  & 0.315 (0.093) & 0.424 (0.135) & 0.371 (0.143) & 1.348 (0.809) \tabularnewline
		\cline{3-7}
		& & 200  & 0.225 (0.084) & 0.341 (0.090) & 0.254 (0.097) & 0.902 (0.513) \tabularnewline
		\hline
		\multirow{4}{*}{II} & \multirow{2}{*}{2} & 50  & 0.556 (0.119) & 0.521 (0.183) & 0.541 (0.160) & 2.061 (1.061) \tabularnewline
		\cline{3-7}
		& & 200  & 0.474 (0.068) & 0.436 (0.132) & 0.465 (0.101) & 1.625 (0.632) \tabularnewline
		\cline{2-7}
		& \multirow{2}{*}{4} & 50  & 0.536 (0.126) & 0.472 (0.148) & 0.517 (0.139) & 2.014 (0.868) \tabularnewline
		\cline{3-7}
		& & 200  & 0.457 (0.063) & 0.419 (0.133) & 0.431 (0.112) & 1.543 (0.604) \tabularnewline
		\hline
		\multirow{4}{*}{III} & \multirow{2}{*}{2} & 50  & 0.503 (0.090) & 0.511 (0.154) & 0.491 (0.130) & 1.248 (0.650) \tabularnewline
		\cline{3-7}
		& & 200  & 0.473 (0.041) & 0.439 (0.092) & 0.366 (0.052) & 1.136 (0.439) \tabularnewline
		\cline{2-7}
		& \multirow{2}{*}{4} & 50  & 0.493 (0.075) & 0.499 (0.120) & 0.487 (0.122) & 1.217 (0.727) \tabularnewline
		\cline{3-7}
		& & 200  & 0.469 (0.055) & 0.423 (0.087) & 0.358 (0.063) & 0.997 (0.316) \tabularnewline
		\hline
	\end{tabular}
\end{center}
	\label{tab:cov-sparse-1}
\end{table}

\section{Application}\label{sec:application}
We applied the proposed method to analyze the longitudinal data that was collected and detailed in \citet{BACHRACH1999}. It consists of longitudinal measurements of spinal bone mineral density for 423 healthy subjects.  The measurement for each individual was observed annually for up to 4 years. Among 423 subjects, we focused on  $n=280$  subjects  ranging in age from 8.8 to 26.2 years who completed at least 2 measurements. A plot for the design of the covariance function is  given in Figure \ref{fig:bone-design}, while a scatter plot for the raw covariance surface is given in Figure \ref{fig:raw-cov}. 
The raw covariance surface seems to decay rapidly to zero as design points move away from the diagonal. This motivated us to estimate the covariance structure with a Mat\'ern correlation function. This method is referred to as SNPTM. In addition, we also used the more flexible $d_n$-Fourier basis family to see whether a better fit can be achieved, where $d_n=2$ was selected by Akaike information criterion (AIC). Such approach is denoted by SNPTF. 

The estimated variance of the measurement error is $1.5\times 10^{-3}$ by the method proposed in Section \ref{sec:sigma0}, $10^{-6}$ by PACE and $7.8\times 10^{-7}$ by LM, respectively.  The estimates of the covariance surface are depicted in Figure \ref{fig:cov}. We observe that, the estimates produced by SNPTM and SNPTF are similar in the diagonal region, while visibly differ in the off-diagonal region. For this dataset, the upward off-diagonal parts of the estimated covariance surface by SNPTF seem artificial, so we recommend the SNPTM estimate for this data. For the PACE estimate, due to the missing data in the off-diagonal region and insufficient observations at two ends of the diagonal region, it suffers from significant boundary effect.

The mean function estimated by SNPTM\footnote{SNPTM, SNPTF and PACE use the same method to estimate the mean function.} shown in the left panel of Figure \ref{fig:mean-var} and found similar to its counterpart in \cite{Lin2019a}, suggests that the spinal bone mineral density increases rapidly from age 9 to age 16, and then slows down afterward. The mineral density has the largest  variation around age 14, indicated by the variance function estimated by SNPTM\footnote{SNPTM and SNPTF use the same method to estimate the variance function.} and shown in the middle panel of Figure \ref{fig:mean-var}. As a comparison, the PACE estimate, shown in the right panel of Figure \ref{fig:mean-var}, suffers from the boundary effect that is passed from the PACE estimate of the covariance function, because the PACE method estimates the variance function by the diagonal of the estimated covariance function.

\begin{figure}
\begin{center}
\includegraphics[width=0.5\textwidth]{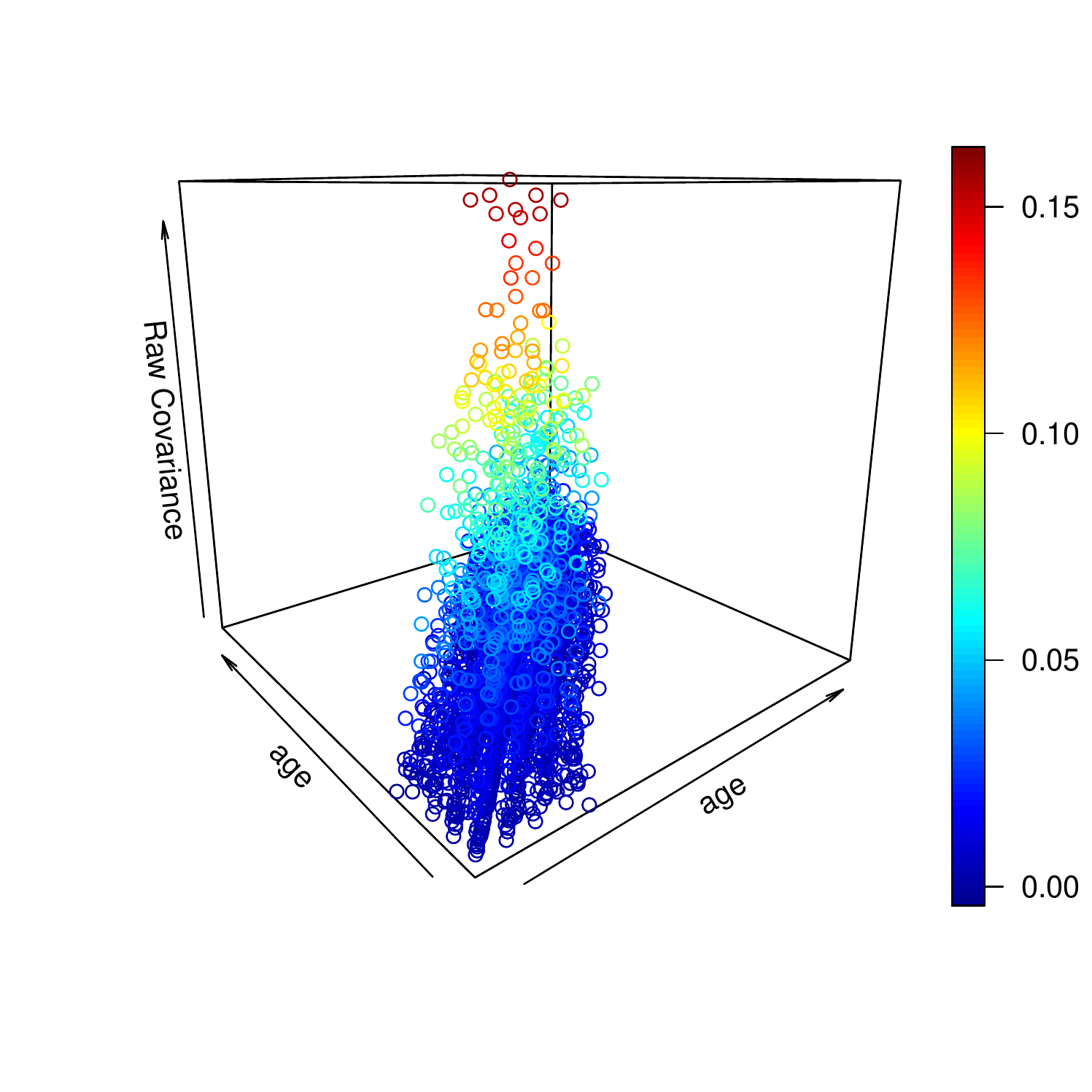}
\end{center}
\vspace{-1cm}
\caption{Scatter plot of the raw covariance function of the spinal bone mineral density data.}
\label{fig:raw-cov}
\end{figure}

\begin{figure}
	\begin{center}
	\begin{minipage}{0.32\textwidth}
			\includegraphics[scale=0.37]{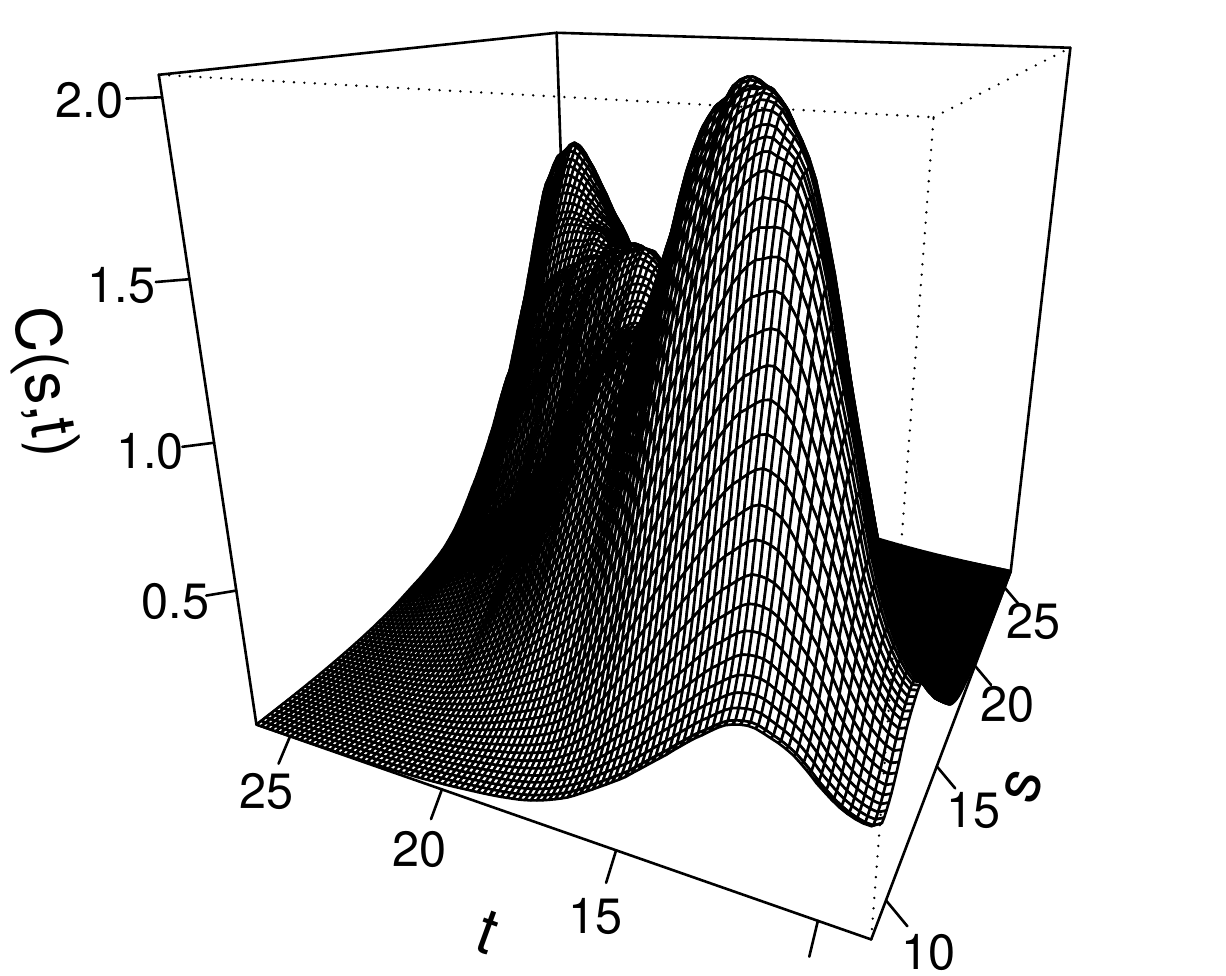}
		\end{minipage}
		\begin{minipage}{0.32\textwidth}
			\includegraphics[scale=0.37]{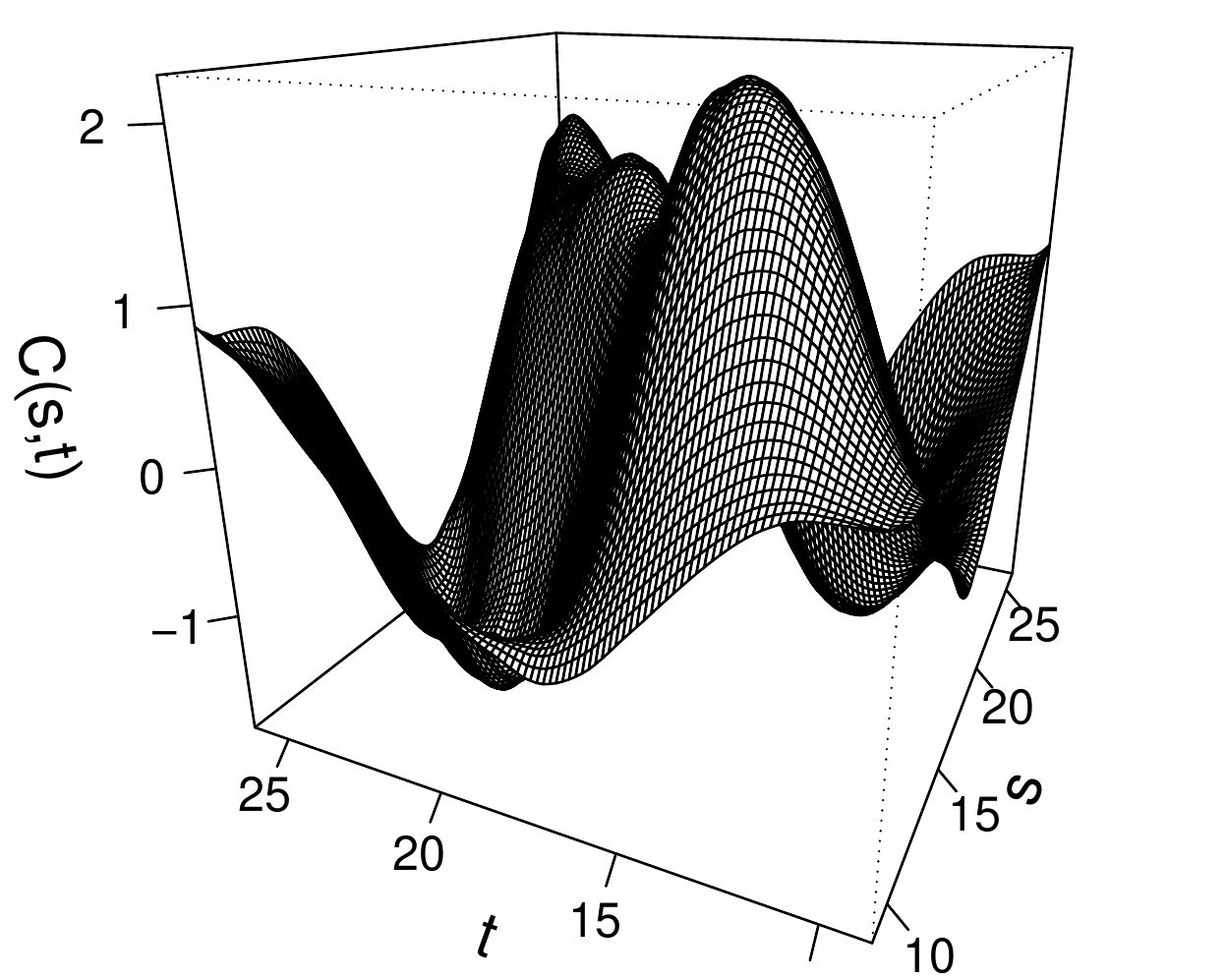}
		\end{minipage}
	\begin{minipage}{0.32\textwidth}
		\includegraphics[scale=0.37]{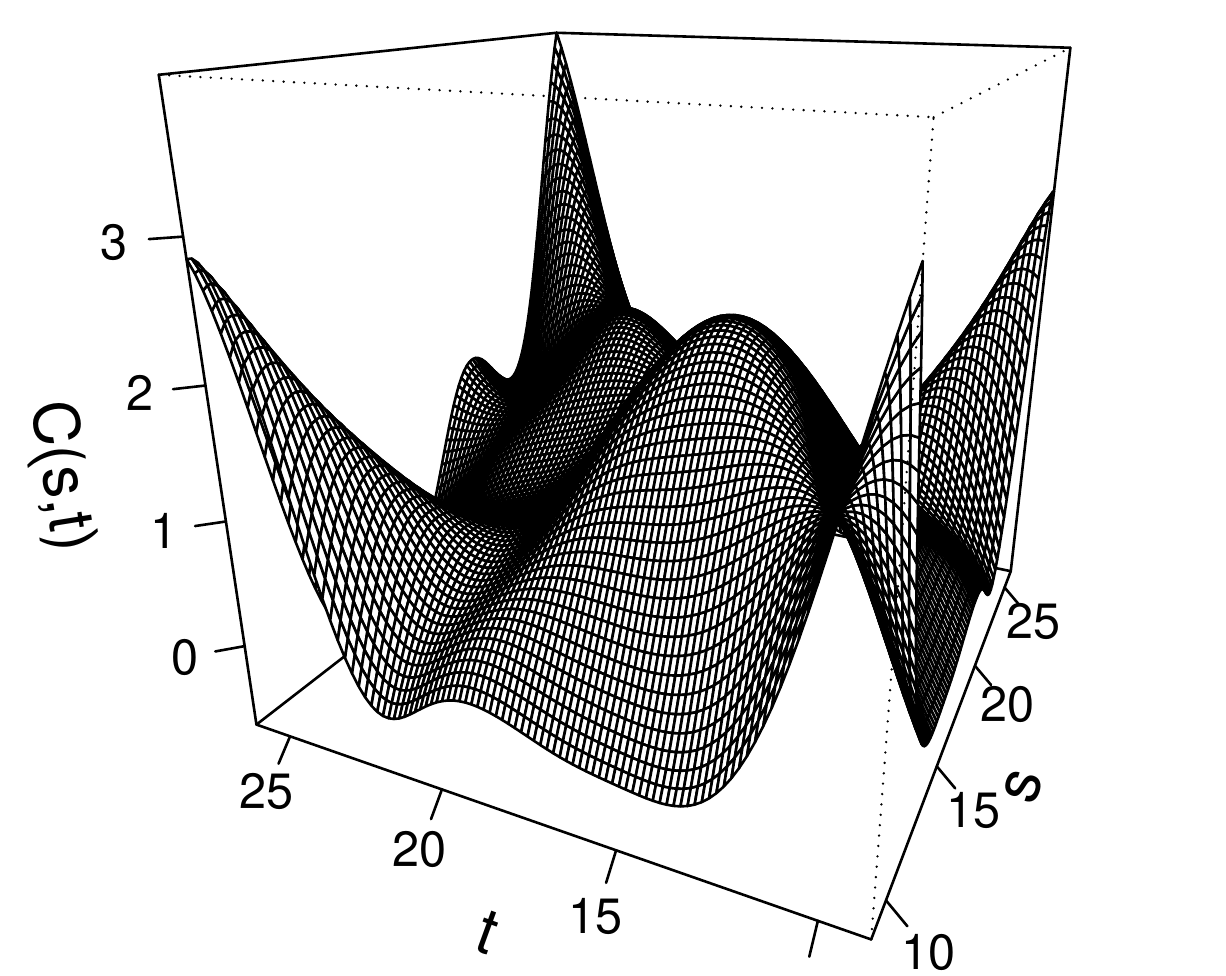}
	\end{minipage}

	\end{center}
	\caption{The estimated covariance functions by SNPTM (left), SNPTF (middle) and PACE (right). The $z$-axis is scaled by $10^{-2}$ for visualization.}
	\label{fig:cov}
\end{figure}

\begin{figure}
\begin{center}
\begin{minipage}{0.32\textwidth}
\includegraphics[scale=0.4]{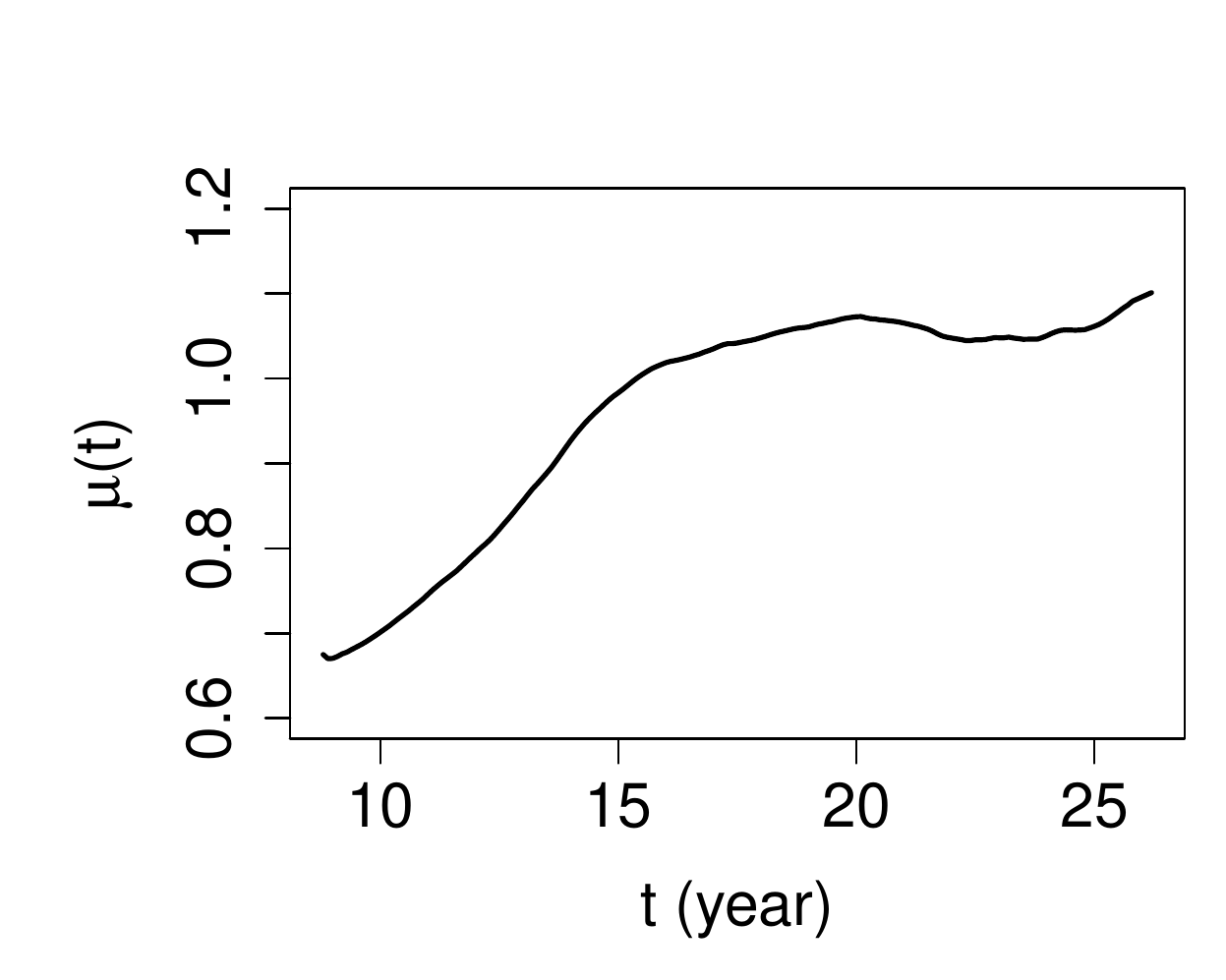}
\end{minipage}
\begin{minipage}{0.32\textwidth}
\includegraphics[scale=0.4]{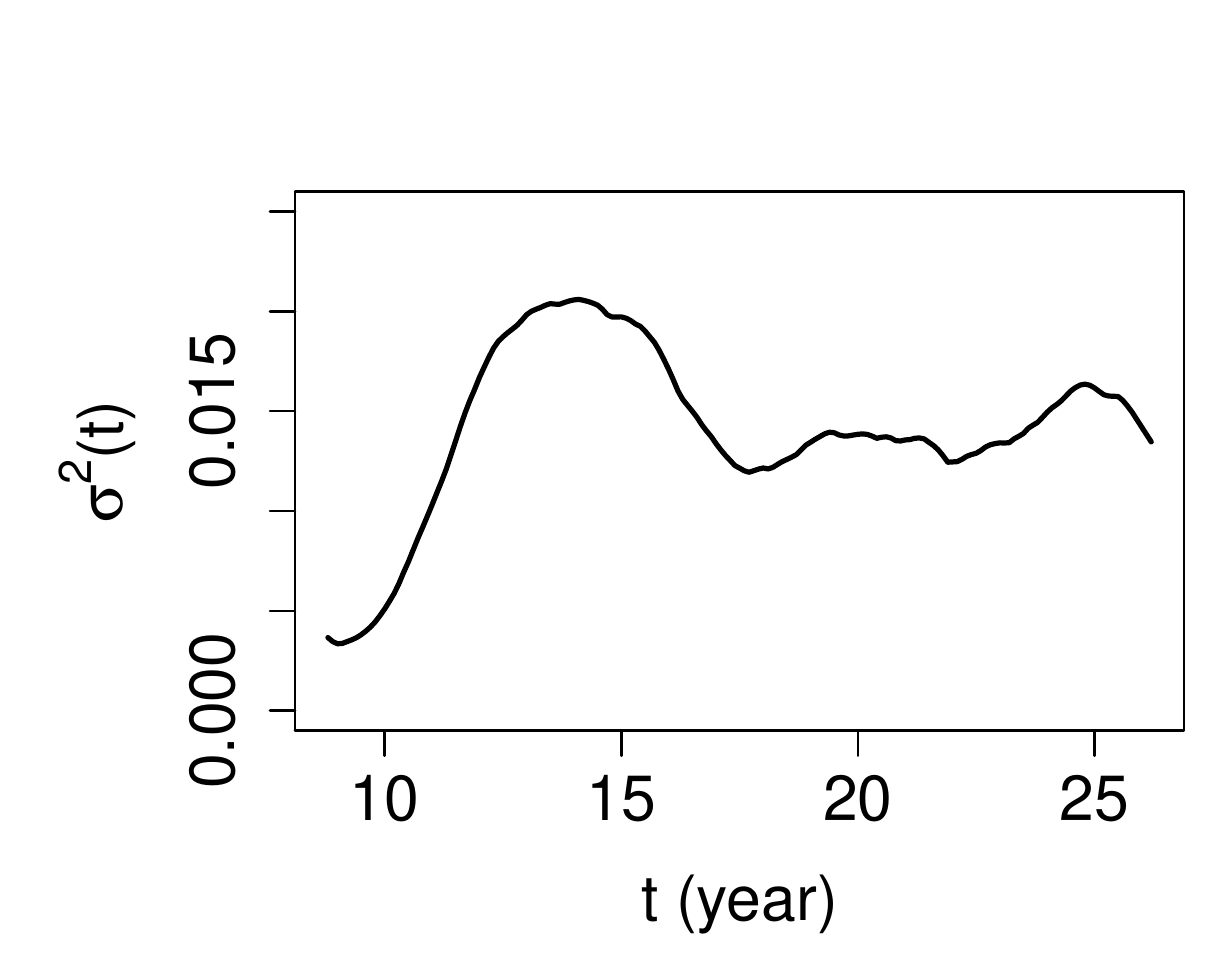}
\end{minipage}
\begin{minipage}{0.32\textwidth}
\includegraphics[scale=0.4]{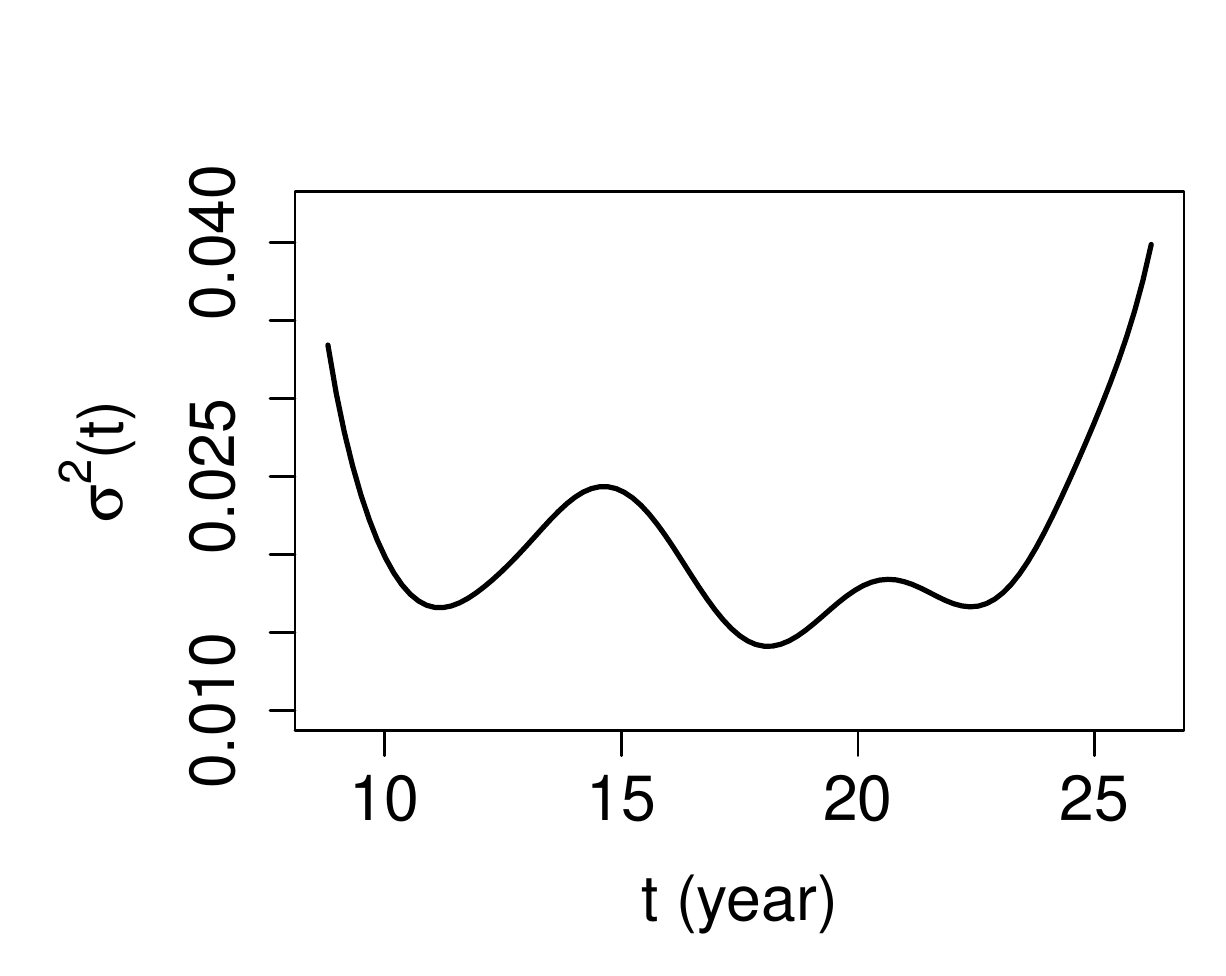}
\end{minipage}
\end{center}
\caption{The estimated mean function (left), the estimated variance function by SNPTM and SNPTF (middle), and the estimated variance function by PACE (right).}
\label{fig:mean-var}
\end{figure}

\section{Concluding Remarks}\label{sec:conclusion}

In this paper, we consider the mean and covariance estimation for functional snippets.  The estimation of the mean function is still an interpolation problem so previous approaches based on local smoothing methods still work, except that the theory needs a little adjustment to reflect the new design of functional snippets.  However, the estimation of the covariance function is quite different because it is now an extrapolation problem rather an interpolation problem, so previous approaches based on local smoothing do not work anymore.  We propose a hybrid approach that   leverages  the available information  and structure of the correlation in the diagonal band to estimate the correlation function parametrically but the variance function nonparametrically. Because the dimension of the parameters can grow with the sample size, the approach is very flexible and can be made nearly nonparametric for the final covariance estimate.   

An interesting feature of the algorithm is that it reverses the order of estimation for the variance components, compared to existing approaches for non-snippets functional data, by first estimating the noise variance $\sigma_0$, then estimating the variance function $\sigma_X^2(t)$, followed by  the fitting of the correlation function.  The estimation of the covariance function is performed at the very end when all other components have been estimated. The proposed approach differs substantially from traditional approaches, such as PACE \citep{Yao2005a}, which estimate the covariance function first, from there  the variance function is obtained as a byproduct through the diagonal elements of the covariance estimate, while the noise variance  is estimated at the very end.  The new procedure to estimate the noise variance is both simpler and better than the PACE estimates.  Thus, even if the data are non-snippet types, one can use the new  method proposed in Section \ref{sec:sigma0} to estimate the noise variance.

We emphasize that, although the proposed method targets functional snippets, it is also applicable to functional fragments or functional data in which each curve consists of multiple disjoint snippets. In addition, the theory presented in Section \ref{sec:theory} can be slightly modified to accommodate such data. In contrast, methods designed for nonsnippet functional data are generally not applicable to functional snippets, due  to the reasons  discussed in Section \ref{sec:introduction}. In practice, one might distinguish between functional snippets and nonsnippets  by the design plot like Figure \ref{fig:bone-design}. If the support points cover the entire region, then the data are of the nonsnippet type.  Otherwise they are functional snippets. However, there might be some case that it is unclear whether the entire region is fully covered by support points, especially when data are sparsely observed. In such situation,  snippet-based methods,  such as the proposed one, is a  safer option.

Reliable estimates of the mean and covariance functions are fundamental to the analysis of functional data.  They are also the building blocks of functional regression methods and functional hypothesis test procedures. 
 The proposed estimators for the mean and covariance of functional snippets together provide a stepping stone to future study on regression and inference that are specific to functional snippets.
 
 \section*{Supplementary Material}
 The online supplementary material contains additional simulation results, as well as information for implementation of the proposed method in the \texttt{R} package \texttt{mcfda}\footnote{https://github.com/linulysses/mcfda.}.

\appendix

\section*{Appendix}
\subsection*{Selection of $h_0$}
The constant $0.29$ in the empirical rule $h_0=0.29\hat\delta\|\hat\varsigma\|_2(nm^2)^{-1/5}$ presented in Section \ref{sec:sigma0} was determined by  optimizing $\sum\{\hat h-c\hat\delta\|\hat\varsigma\|_2(nm^2)^{-1/5}\}^2$ over $c\in\real$, where the summation is taken over the combinations of various parameters. Specifically, for each tuple $(n,m,\delta,\sigma_{0}^2,\covarop)$,  
we generated a batch of $G=100$ independent datasets of $n$ centered Gaussian snippets with the covariance function $\covarop$. Each snippet was recorded at $m$ random points from a random subinterval of length $\delta$ in $[0,1]$. For each batch of datasets, we found $\hat h$ to minimize $\sum_{r=1}^{G} \{\hat\sigma_{0,r}^2(\hat h)-\sigma_{0}^2\}^2$, where $\hat\sigma_{0,r}^2(\hat h)$ is the estimate of $\sigma_{0}^2$ based the $r$th dataset in the batch and by using the proposed method with the bandwidth $\hat h$. We also obtained the quantities $\hat\delta=G^{-1}\sum_{r=1}^G \hat\delta_{r}$ and $\|\hat\varsigma\|_2=G^{-1}\sum_{r=1}^G \|\hat\varsigma_r\|_2$, where $\hat\delta_r$ and $\hat\varsigma_r$ are the estimate of $\delta$ and $\varsigma$ based on the $r$th dataset in the batch, respectively. In this way, we obtain a collection $\mathscr H$ of vectors $(\hat h,n,m,\hat\delta,\|\hat\varsigma\|_2)$. Finally, we found $c=0.29$ to minimize $\sum\{\hat h-c\hat\delta\|\hat\varsigma\|_2(nm^2)^{-1/5}\}^2$, where the summation is taken over the collection $\mathscr H$.

In the above process, the covariance function $\covarop$ was taken from a collection composed by 1) covariance functions whose correlation part is the correlation function listed in Section \ref{sec:method} with various values of the parameters and whose variance functions are exponential functions, squared sin/cos functions and positive polynomials, 2) covariance functions  $\covarop(s,t)=a\min\{s,t\}$ with various values of $a>0$, 3)  covariance functions $\covarop(s,t)=\sum_{k=1}^K ak^{-\lambda}\phi_k(s)\phi_k(t)$ with various values of $a>0$, $\lambda>0$ and $K\geq 1$, where the functions $\phi_k$ are the Fourier basis functions described in Section \ref{sec:simulation}, and 4) covariance functions $\covarop(s,t)=\sum_{1\leq j,k\leq K} ae^{-b|j-k|}$  with various choices of $a>0$, $b>0$ and $K\geq 1$.

\subsection*{Technical Lemmas}

\begin{lemma}\label{lem:A0-A1-A2-B}~
\begin{enumerate}
[leftmargin=*,labelsep=2mm,label= \textup{(\alph*)}]
\setcounter{enumi}{0}%
\item\label{lem:enum:A1} Under conditions \ref{cond:A1}--\ref{cond:A2}, one has $A_{2}=A_{1}+O(h_0^{3})$.
\item\label{lem:enum:B}With condition \ref{cond:A1}, 
$\expect(\hat{B}-B)^{2}=O(n^{-1}m^{-2}h_0+n^{-1}m^{-1}h_0^{2})$.
\item\label{lem:enum:A0-A1}Under conditions \ref{cond:A1}--\ref{cond:A3}, $\expect\{(\hat{A}_0-\hat{A}_1)-(A_0-A_1)\}^{2}=O(h_0^6+n^{-1}m^{-2}h_0+n^{-1}h_0^{2})$. If $\expect L_X^4<\infty$ is also assumed, then $\expect\{(\hat{A}_0-\hat{A}_1)-(A_0-A_1)\}^{2}=O(h_0^6+n^{-1}m^{-2}h_0+n^{-1}m^{-1}h_0^{2})$.
\end{enumerate}
\end{lemma}
\begin{proof}
To show $A_2=A_1+O(h_0^3)$ in part \ref{lem:enum:A1}, we define $\tdomain_{h_0,\delta}=\{(s,t,u):u\in[\delta/2,1-\delta/2],\:u-\delta/2\leq s,t\leq u+\delta/2,\:|s-t|<h_0\}$ and $g(s,t,u)=\{\covarop(s,t)+\mu(s)\mu(t)\}f_{T|O}(s|u)f_{T|O}(t|u)f_{O}(u)$.
Let $g_{s}$ be the partial derivative of $g$ with respect to $s$.
Then, $g_{s}$ is Lipschitz continuous given condition \ref{cond:A1} and \ref{cond:A2}. With $t^{\ast}$ denoting a real number satisfying $\min(s,t)\leq t^{\ast}\leq\max(s,t)$,
one has
\begin{align*}
A_{2} & =\iiint_{\tdomain_{h_0,\delta}}[g(t,t,u)+g_{s}(t,t,u)(s-t)+\{g_{s}(t^{\ast},t,u)-g_{s}(t,t,u)\}(s-t)^{2}\}]\diffop s\diffop t\diffop u\\
 & =A_{1}+\iiint_{\tdomain_{h_0,\delta}}g_{s}(t,t,u)(s-t)\diffop s\diffop t\diffop u+O(h_0^{3}) =A_{1}+O(h_0^{3}),
\end{align*}
where the last equality is obtained by observing that 
\begin{align*}
\iiint_{\tdomain_{h_0,\delta}}g_{s}(t,t,u)(s-t)\diffop s\diffop t\diffop u= & \int_{\delta/2}^{1-\delta/2}\int_{u-\delta/2+h_0}^{u+\delta/2-h_0}\int_{t-h_0}^{t+h_0}g_{s}(t,t,u)(s-t)\diffop s\diffop t\diffop u\\
 & +\int_{\delta/2}^{1-\delta/2}\int_{u-\delta/2}^{u-\delta/2+h_0}\int_{\max(u-\delta/2,t-h_0)}^{\min(u+\delta/2,t+h_0)}g_{s}(t,t,u)(s-t)\diffop s\diffop t\diffop u\\
 & +\int_{\delta/2}^{1-\delta/2}\int_{u+\delta/2-h_0}^{u+\delta/2}\int_{\max(u-\delta/2,t-h_0)}^{\min(u+\delta/2,t+h_0)}g_{s}(t,t,u)(s-t)\diffop s\diffop t\diffop u\\
= & 0+O(h_0^{3})+O(h_0^{3})=O(h_0^{3}).
\end{align*}

For part \ref{lem:enum:B}, it is seen that $\expect\hat{B}=B$ and
\begin{align}
\expect(\hat{B}-B)^{2} & =\expect\left[\frac{1}{n}\sum_{i=1}^{n}\frac{1}{m(m-1)}\sum_{j\neq l}1_{|T_{ij}-T_{il}|<h_0}-B\right]^{2}\nonumber\\
& =\frac{1}{n}\expect\left[\frac{1}{m(m-1)}\sum_{j\neq l}1_{|T_{ij}-T_{il}|<h_0}-B\right]^{2}.\label{eq:pf-lem-1}
\end{align}
Now we first observe that $\expect(1_{|T_{ij}-T_{il}|<h_{0}}\mid O_{i})=B$,
since 
\begin{align*}
\expect(1_{|T_{ij}-T_{il}|<h_{0}}\mid O_{i}) & =\iint_{\stackrel{|s-t|<h_{0}}{O_{i}-\delta/2\leq s,t\leq O_{i}+\delta/2}}f_{T|O}(s|O_{i})f_{T|O}(t|O_{i})\diffop s\diffop t\\
 & =\iint_{\stackrel{|s-t|<h_{0}}{O_{i}-\delta/2\leq s,t\leq O_{i}+\delta/2}}f_{0}(s-O_{i}+\delta/2)f_{0}(t-O_{i}+\delta/2)\diffop s\diffop t\\
 & =\iint_{\stackrel{|s-t|<h_{0}}{0\leq s,t\leq\delta}}f_{0}(s)f_{0}(t)\diffop s\diffop t
\end{align*}
and 
\begin{align*}
B=\expect1_{|T_{ij}-T_{il}|<h_{0}} & =\expect\expect(1_{|T_{ij}-T_{il}|<h_{0}}\mid O_{i})=\iint_{\stackrel{|s-t|<h_{0}}{0\leq s,t\leq\delta}}f_{0}(s)f_{0}(t)\diffop s\diffop t.
\end{align*}
Therefore, if $j,l,p,q$ are all distinct, then
\begin{align*}
 & \expect\{(1_{|T_{ij}-T_{il}|<h_{0}}-B)(1_{|T_{ip}-T_{iq}|<h_{0}}-B)\}\\
 & =\expect\expect\{(1_{|T_{ij}-T_{il}|<h_{0}}-B)(1_{|T_{ip}-T_{iq}|<h_{0}}-B)\mid O_{i}\}\\
 & =\expect\{\expect(1_{|T_{ij}-T_{il}|<h_{0}}-B\mid O_{i})\expect(1_{|T_{ip}-T_{iq}|<h_{0}}-B\mid O_{i})\} =0.
\end{align*}
It is relatively straightforward to show that if $j=p$ but $l\neq q$
or $j=q$ but $l\neq p$, then $\expect\{(1_{|T_{ij}-T_{il}|<h_{0}}-B)(1_{|T_{ip}-T_{iq}|<h_{0}}-B)\}=O(h_{0}^{2})$,
and if $j=p$ and $l=q$ or $j=q$ and $l=p$, then $\expect\{(1_{|T_{ij}-T_{il}|<h_{0}}-B)(1_{|T_{ip}-T_{iq}|<h_{0}}-B)\}=O(h_{0})$.
Assembling the above results, one has
\begin{align*}
\expect\left[\frac{1}{m(m-1)}\sum_{j\neq l}1_{|T_{ij}-T_{il}|<h_{0}}-B\right]^{2} & =O(m^{-2}h_{0}+m^{-1}h_{0}^{2}),
\end{align*}
which together with \eqref{eq:pf-lem-1} implies the conclusion of part \ref{lem:enum:B}.

For part \ref{lem:enum:A0-A1}, with the aid of part \ref{lem:enum:A1}, it is straightforward to see that \begin{equation}\label{eq:pf-lem-bias}
\expect\{(\hat{A}_0-\hat{A}_1)-(A_0-A_1)\}=O(h_0^3).
\end{equation}
Now we shall calculate the variance of $\hat{A}_0-\hat{A}_1$. With definition $E_{0} =\expect(Y_{ij}-Y_{il})^{2}1_{|T_{ij}-T_{il}|<h_{0}}$, one derives
\begin{align}
& \mathrm{Var}(\hat{A}_{0}-\hat{A}_{1}) \nonumber \\ \nonumber
& =\mathrm{Var}\left(\frac{1}{n}\sum_{i=1}^{n}\frac{1}{m(m-1)}\sum_{j\neq l}\frac{(Y_{ij}-Y_{il})^{2}}{2}1_{|T_{ij}-T_{il}|<h_{0}}\right)\\ \nonumber
 & =\frac{1}{4n}\mathrm{Var}\left(\frac{1}{m(m-1)}\sum_{j\neq l}(Y_{ij}-Y_{il})^{2}1_{|T_{ij}-T_{il}|<h_{0}}\right)\\
 & =\frac{1}{4n}\left(\frac{1}{m^{2}(m-1)^{2}}\sum_{j\neq l}\sum_{p\neq q}\expect\{(Y_{ij}-Y_{il})^{2}1_{|T_{ij}-T_{il}|<h_{0}}-E_{0}\}\{(Y_{ip}-Y_{iq})^{2}1_{|T_{ip}-T_{iq}|<h_{0}}-E_{0}\}\right)\nonumber\\
 & \equiv\frac{1}{4n}\left(\frac{1}{m^{2}(m-1)^{2}}\sum_{j\neq l}\sum_{p\neq q}V(j,l,p,q)\right).\label{eq:pf-lem-3}
\end{align}
Below we derive bounds for the term $V(j,l,p,q)$. \begin{itemize}
\item Case 1: $j$, $l$, $p$ and $q$ are all distinct. In this case,
via straightforward computation, one can show that
$V(j,l,p,q) =\expect\{(Y_{ij}-Y_{il})^{2}1_{|T_{ij}-T_{il}|<h_{0}}\}\{(Y_{ip}-Y_{iq})^{2}1_{|T_{ip}-T_{iq}|<h_{0}}\}-E_{0}^{2}=O(h_{0}^{2})$.
\item Case 2: $j=p$ but $l\neq q$ or $j=q$ but $l\neq p$. Similar
to Case 1, one has $V(j,l,p,q)=O(h_{0}^{2})$.
\item Case 3: $j=p$ and $l=q$ or $j=q$ and $l=p$. In this case, 
\[
V(j,l,p,q)=\expect\{(Y_{ij}-Y_{il})^{4}1_{|T_{ij}-T_{il}|<h_{0}}\}-E_{0}^{2}=O(h_{0}).
\]
\end{itemize}
Based on the above bounds, we have $\mathrm{Var}(\hat{A}_{0}-\hat{A}_{1})=O(n^{-1}h_0^2+n^{-1}m^{-1}h_0^2+n^{-1}m^{-2}h_0)=O(n^{-1}h_0^2+n^{-1}m^{-2}h_0)$. Together with the bias given in \eqref{eq:pf-lem-bias}, this implies the first statement of part \ref{lem:enum:A0-A1}.

For the second statement of part \ref{lem:enum:A0-A1}, we observe that with condition $\expect L_X^4<\infty$, the bound in Case 1 can be sharpened in the following way. First, we see that \begin{align*}
E_{0} & =\expect\{X_{i}(T_{ij})-X_{i}(T_{il})\}^{2}1_{|T_{ij}-T_{il}|<h_{0}}+\expect(\varepsilon_{ij}-\varepsilon_{il})^{2}1_{|T_{ij}-T_{il}|<h_{0}} =E_{1}+2\sigma_{0}^{2}B,
\end{align*}
where $E_{1}=\expect\{X_{i}(T_{ij})-X_{i}(T_{il})\}^{2}1_{|T_{ij}-T_{il}|<h_{0}}$. Then,
we decompose $V(j,l,p,q)$ into $I_1+I_2+I_3+I_4$, where 
\begin{align*}
I_{1} & =\expect[\{X(T_{ij})-X(T_{il})\}^{2}1_{|T_{ij}-T_{il}|<h_{0}}-E_{1}][\{X(T_{ip})-X(T_{iq})\}^{2}1_{|T_{ip}-T_{iq}|<h_{0}}-E_{1}],\\
I_{2} & =\expect[\{X(T_{ij})-X(T_{il})\}^{2}1_{|T_{ij}-T_{il}|<h_{0}}-E_{1}][(\varepsilon_{ip}-\varepsilon_{iq})^{2}1_{|T_{ip}-T_{iq}|<h_{0}}-2\sigma_{0}^{2}B],\\
I_{3} & =\expect[(\varepsilon_{ij}-\varepsilon_{il})^{2}1_{|T_{ij}-T_{il}|<h_{0}}-2\sigma_{0}^{2}B][\{X(T_{ip})-X(T_{iq})\}^{2}1_{|T_{ip}-T_{iq}|<h_{0}}-E_{1}],\\
I_{4} & =\expect[(\varepsilon_{ij}-\varepsilon_{il})^{2}1_{|T_{ij}-T_{il}|<h_{0}}-2\sigma_{0}^{2}B][(\varepsilon_{ip}-\varepsilon_{iq})^{2}1_{|T_{ip}-T_{iq}|<h_{0}}-2\sigma_{0}^{2}B].
\end{align*}
For $I_{2}$, one can show that 
\begin{align*}
I_{2} & =\expect\expect\left([\{X(T_{ij})-X(T_{il})\}^{2}1_{|T_{ij}-T_{il}|<h_{0}}-E_{1}][(\varepsilon_{ip}-\varepsilon_{iq})^{2}1_{|T_{ip}-T_{iq}|<h_{0}}-2\sigma_{0}^{2}B]\mid O_{i}\right)\\
 & =\expect\left(\expect[\{X(T_{ij})-X(T_{il})\}^{2}1_{|T_{ij}-T_{il}|<h_{0}}-E_{1}\mid O_{i}]\expect[(\varepsilon_{ip}-\varepsilon_{iq})^{2}1_{|T_{ip}-T_{iq}|<h_{0}}-2\sigma_{0}^{2}B\mid O_{i}]\right)\\
 & =0,
\end{align*}
where the first equality is due to the assumption that $T_{i1},\ldots,T_{im}$
are i.i.d. conditional on $O_{i}$, and the second one is based on
the following observation
\begin{align*}
\expect[(\varepsilon_{ip}-\varepsilon_{iq})^{2}1_{|T_{ip}-T_{iq}|<h_{0}}-2\sigma_{0}^{2}B\mid O_{i}] & =2\sigma_{0}^{2}\expect(1_{|T_{ip}-T_{iq}|<h_{0}}\mid O_{i})-2\sigma_{0}^{2}B =2\sigma_{0}^{2}B-2\sigma_{0}^{2}B=0,
\end{align*}
where we recall that $\expect(1_{|T_{ij}-T_{il}|<h_{0}}\mid O_{i})=B$.
Similarly, $I_{3}=0$ and $I_{4}=0$. For $I_{1}$, one can show that
\begin{align*}
|I_{1}| & =|\expect[\{X(T_{ij})-X(T_{il})\}^{2}1_{|T_{ij}-T_{il}|<h_{0}}-E_{1}][\{X(T_{ip})-X(T_{iq})\}^{2}1_{|T_{ip}-T_{iq}|<h_{0}}-E_{1}]|\\
 & =|\expect[\{X(T_{ij})-X(T_{il})\}^{2}1_{|T_{ij}-T_{il}|<h_{0}}\{X(T_{ip})-X(T_{iq})\}^{2}1_{|T_{ip}-T_{iq}|<h_{0}}]-E_{1}^{2}|\\
 & \leq\expect(L_{X}^{4}|T_{ij}-T_{il}|^{2}|T_{ip}-T_{iq}|^{2}1_{|T_{ij}-T_{il}|<h_{0}}1_{|T_{ip}-T_{iq}|<h_{0}})+E_{1}^{2}\\
 & \leq h_{0}^{4}\expect L_{X}^{4}\expect1_{|T_{ij}-T_{il}|<h_{0}}1_{|T_{ip}-T_{iq}|<h_{0}}+E_{1}^{2}\\
 & =O(h_{0}^{6})+E_{1}^{2},
\end{align*}
where the first inequality is due to the Lipschitz continuity property of sample paths. Again, based 
on such continuity property, one has $E_{1}=\expect\{X_{i}(T_{ij})-X_{i}(T_{il})\}^{2}1_{|T_{ij}-T_{il}|<h_{0}}\leq\expect L_{X}^{2}|T_{ij}-T_{il}|^{2}1_{|T_{ij}-T_{il}|<h_{0}}\leq h_{0}^{2}\expect L_{X}^{2}\expect1_{|T_{ij}-T_{il}|<h_{0}}=O(h_{0}^{3})$.
Therefore, we conclude that $I_{1}=O(h_{0}^{6})$. Together with $I_2=I_3=I_4=0$, this implies that $V(j,l,p,q)=O(h_{0}^{6})$. It further indicates that $\mathrm{Var}(\hat{A}_{0}-\hat{A}_{1})=O(n^{-1}h_0^6+n^{-1}m^{-1}h_0^2+n^{-1}m^{-2}h_0)$. Combined with the bias term in \eqref{eq:pf-lem-bias}, this implies the second statement of part \ref{lem:enum:A0-A1}. 
\end{proof}

\subsection*{Proofs of Main Results}

\begin{proof}[Proof of Proposition \ref{thm:theta}]
 For the moment, we assume $\mu\equiv0$. Denote
\begin{align*}
Q_{n}(\theta) & =\frac{1}{n}\sum_{i=1}^{n}\frac{1}{m(m-1)}\sum_{1\leq j\neq l\leq m}\{\sigma_X(T_{ij})\sigma_X(T_{il})\rho_{\theta}(T_{ij},T_{il})-C_{ijl}\}^{2}.
\end{align*}
Now we show that 
\begin{equation}
\left\Vert \frac{\partial\hat{Q}_{n}}{\partial\theta}-\frac{\partial Q_{n}}{\partial\theta}\right\Vert =\Op\left(\sqrt{\frac{d_{n}a_{n}\log n}{n}}\right),\label{eq:score-function}
\end{equation}
where $a_{n}=(\log n)\{(nm)^{-4/5}+n^{-1}\}$. First, we observe that
\begin{align*}
\frac{\partial\hat{Q}_{n}}{\partial\theta}-\frac{\partial Q_{n}}{\partial\theta} & =I_{1}+I_{2}+I_{3}
\end{align*}
with 
\begin{align*}
I_{1}= & \frac{1}{n}\sum_{i=1}^{n}\frac{1}{m(m-1)}\sum_{1\leq j\neq l\leq m}2\{\sigma_X(T_{ij})\sigma_X(T_{il})\rho_{\theta}(T_{ij},T_{il})-C_{ijl}\}\times\\
& \{\hat{\sigma}_X(T_{ij})\hat{\sigma}_X(T_{il})-\sigma_X(T_{ij})\sigma_X(T_{il})\}\frac{\partial\rho_{\theta}(T_{ij},T_{il})}{\partial\theta},\\
I_{2}= & \frac{1}{n}\sum_{i=1}^{n}\frac{1}{m(m-1)}\sum_{1\leq j\neq l\leq m}2\{\hat{\sigma}_X(T_{ij})\hat{\sigma}_X(T_{il})-\sigma_X(T_{ij})\sigma_X(T_{il})\}\rho_{\theta}(T_{ij},T_{il})\times\\
& \sigma_X(T_{ij})\sigma_X(T_{il})\frac{\partial\rho_{\theta}(T_{ij},T_{il})}{\partial\theta},\\
I_{3}= & \frac{1}{n}\sum_{i=1}^{n}\frac{1}{m(m-1)}\sum_{1\leq j\neq l\leq m}2\{\hat{\sigma}_X(T_{ij})\hat{\sigma}_X(T_{il})-\sigma_X(T_{ij})\sigma_X(T_{il})\}\rho_{\theta}(T_{ij},T_{il})\times\\
 & \{\hat{\sigma}_X(T_{ij})\hat{\sigma}_X(T_{il})-\sigma_X(T_{ij})\sigma_X(T_{il})\}\frac{\partial\rho_{\theta}(T_{ij},T_{il})}{\partial\theta}.
\end{align*}
To derive the rate for $I_{1}$, we define
\begin{align*}
G & =\frac{1}{n}\sum_{i=1}^{n}\frac{1}{m(m-1)}\sum_{1\leq j\neq l\leq m}2\{\sigma_X(T_{ij})\sigma_X(T_{il})\rho_{\theta}(T_{ij},T_{il})-C_{ijl}\}\equiv\frac{1}{n}\sum_{i=1}^{n}G_{i}.
\end{align*}
It can be verified that $\expect G_{i}=0$, and also $\expect G_{i}^{2}<\infty$
given condition \ref{cond:A3} and \ref{cond:B2}. We view each $G_{i}$
as a random linear functional from the space 
$\Lambda_{0}=\{f\in C^{2}(\tdomain):\,\|f\|_{\infty}\leq1\}$,
i.e., 
\[
G_{i}(f)\mapsto\frac{1}{m(m-1)}\sum_{1\leq j\neq l\leq m}2\{\sigma_X(T_{ij})\sigma_X(T_{il})\rho_{\theta}(T_{ij},T_{il})-C_{ijl}\}f(T_{ij},T_{il}),
\]
where $f\in\Lambda_{0}$. Then we follow the same lines of the argument
for Lemma 2 of \citet{Severini1992} to establish that $\sqrt{n}G$
converges to a Gaussian element on the Banach space $C(\Lambda_{0})$
of continuous functions on $\Lambda_{0}$ with the sup norm. On the
other hand, using the same technique of \citet{Zhang2016} for the
uniform convergence of the local linear estimator for the mean function,
we can show that $\sup_{t}|\hat{\sigma}_X(t)-\sigma_X(t)|=\Op(\sqrt{a_{n}})$,
and hence $\sup_{s,t}|\hat{\sigma}_X(s)\hat{\sigma}_X(t)-\sigma_X(s)\sigma_X(t)|=\Op(\sqrt{a_{n}})$.
By condition \ref{cond:B2} that $\partial\rho_{\theta}(s,t)/\partial\theta_{j}$
is uniformly bounded for all $j$, we can deduce that, for sufficiently
large $n$, with probability tending to one, the function $(a_{n}\log n)^{-1/2}f_{j}$
with $f_{j}:(s,t)\mapsto\{\hat{\sigma}_X(s)\hat{\sigma}_X(t)-\sigma_X(s)\sigma_X(t)\}\partial\rho_{\theta}(s,t)/\partial\theta_{j}$
falls into $\Lambda_{0}$ for all $j$. Therefore,
\[
\left\Vert \sqrt{n}G\left(\frac{f_{j}}{\sqrt{a_{n}\log n}}\right)\right\Vert \leq\|\sqrt{n}G\|\left\Vert \frac{f_{j}}{\sqrt{a_{n}\log n}}\right\Vert =\Op(1),
\]
where $\Op$ is uniform for all $j$. Noting that $I_{1}=(Gf_{1},\ldots,Gf_{d_{n}})^{\transpose}$,
one can deduce from the above that 
\[
\|I_{1}\|\leq\sqrt{\sum_{j=1}^{d_{n}}\|Gf_{j}\|^{2}}\leq\sqrt{d_{n}}\max_{1\leq j\leq d_{n}}\|Gf_{j}\|=\Op\left(\sqrt{\frac{d_{n}a_{n}\log n}{n}}\right).
\]
When $\mu\neq0$, an argument similar to the above can also be applied
to handle extra terms induced by the discrepancy between $\hat{\mu}$
and $\mu$, so that we still obtain the same rate as the above. Similar
argument applies to $I_{2}$, and we have $I_{2}=\Op(\sqrt{d_{n}a_{n}\log n}/\sqrt{n})$.
It is easy to see that $I_{3}$ is dominated by the other terms. Together,
we establish (\ref{eq:score-function}). It is seen that $\|\partial Q_{n}/\partial\theta\mid_{\theta=\theta_{0}}\|=\Op(\sqrt{d_{n}/n})$.
Thus, we have 
\begin{align*}
\left\Vert \frac{\partial\hat{Q}_{n}}{\partial\theta}\mid_{\theta=\theta_{0}}\right\Vert & \leq\left\Vert \frac{\partial Q_{n}}{\partial\theta}\mid_{\theta=\theta_{0}}\right\Vert +\left\Vert \left(\frac{\partial\hat{Q}_{n}}{\partial\theta}-\frac{\partial Q_{n}}{\partial\theta}\right)\mid_{\theta=\theta_{0}}\right\Vert\\ & =\Op\left(\sqrt{\frac{d_{n}}{n}}+\sqrt{\frac{d_{n}a_{n}\log n}{n}}\right)=\Op\left(\sqrt{\frac{d_{n}}{n}}\right),
\end{align*}
Straightforward but somewhat tedious calculation can show that
\[
\left\Vert \frac{\partial^{2}\hat{Q}_{n}}{\partial\theta^{2}}\mid_{\theta=\theta_{0}}-\frac{\partial^{2}Q}{\partial\theta^{2}}\mid_{\theta=\theta_{0}}\right\Vert =\Op\left(\frac{d_{n}}{\sqrt{n}}+d_{n}\sqrt{a_{n}}\right)=\Op\left(d_{n}\sqrt{a_{n}}\right)
\]
and
\[
\sup_{\theta}\left|\sum_{|\alpha|=3}v^{\alpha}\frac{\partial^{\alpha}\hat{Q}_{n}(\theta)}{\alpha!}\right|=\Op\left(d_{n}^{3/2}\|v\|^{3}\right).
\]

Now let $\eta_{n}=\sqrt{d_{n}^{1+2\tau}/n}$. By Taylor expansion,
\begin{align*}
D(u) & \equiv\hat{Q}_{n}(\theta_{0}+\eta_{n}u)-\hat{Q}_{n}(\theta_{0})\\
 & =\eta_{n}\left(\frac{\partial\hat{Q}_{n}}{\partial\theta}\mid_{\theta=\theta_{0}}\right)^{\transpose}u+\eta_{n}^{2}u^{\transpose}\left(\frac{\partial^{2}\hat{Q}_{n}}{\partial\theta^{2}}\mid_{\theta=\theta_{0}}\right)u+\eta_{n}^{3}\sum_{|\alpha|=3}u^{\alpha}\frac{\partial^{\alpha}\hat{Q}_{n}}{\alpha!}\mid_{\theta=\theta^{\ast}}\\
 & =\Op\left(\eta_{n}\sqrt{\frac{d_{n}}{n}}\right)\|u\|+\eta_{n}^{2}\lambda_{\min}\left(\frac{\partial^{2}Q}{\partial\theta^{2}}\mid_{\theta=\theta_{0}}\right)\|u\|^{2}+\Op\left(\eta_{n}^{3}d_{n}^{3/2}\right)\|u\|^{3}\\
 & \geq \Op\left(d_{n}^{1+\tau}n^{-1}\right)\|u\|+c_{0}d^{1+\tau}n^{-1}\|u\|^{2}+\op(d^{1+\tau}n^{-1})\|u\|^{3} >0
\end{align*}
for some constant $c_0>0$ and if $\|u\|=c$ for a sufficiently large absolute constant $c>0$. Thus,
$\|\hat{\theta}-\theta_{0}\|=\Op(\eta_{n})=\Op(n^{-1/2}d_{n}^{\tau+1/2})$.
\end{proof}

\bibliographystyle{asa}
\bibliography{snippet}


\newpage

\begin{center}
	{\large\bf Supplementary Material to ``Mean and Covariance Estimation for Functional Snippets''}
\end{center}

\renewcommand\thesection{S.\arabic{section}}
\renewcommand\thetable{S.\arabic{table}}
\setcounter{table}{0}
\spacingset{1.25} 

\section{Implementation}
The proposed method has been implemented in the \texttt{R} package \texttt{mcfda}\footnote{https://github.com/linulysses/mcfda.} for mean and covariance estimation in functional data analysis. To use the package, first apply the following command

\texttt{devtools::install\_github("linulysses/mcfda")} 

\noindent
to install the package. For illustration, we use the following command 

\texttt{D <- synfd::sparse.fd(0, synfd::gaussian.process(), n=100, m=5, delta=0.5)}

\noindent
from the \texttt{synfd}\footnote{https://github.com/linulysses/synfd. Installation of this package can be done by the command \texttt{devtools::install\_github("linulysses/synfd")}.} package to synthesize a snippet sample of size $n=100$ in which on average each snippet is observed at $m=5$ random points on $[0,1]$ and the span of each snippet is no larger than $\delta=0.5$. Now call 

\texttt{cov.obj <- covfunc(D\$t, D\$y, method="SP")}

\noindent
with \texttt{method="SP"} to estimate the covariance function by the proposed method with a default setting in which Mat\'ern correlation function is used. A customized correlation function instead of the default one can be adopted; see the package manual for details. Finally, call

\texttt{cov.hat <- predict(cov.obj, seq(0,1,0.01))}

\noindent
to obtain the estimated covariance function in the grid $\{(s,t):s,t=0,0.01,\ldots,1\}$.

\section{Additional simulation results for $\sigma_0^2$}

\begin{table}
	\renewcommand{\arraystretch}{0.8}
	\caption{RMSE and their standard errors for $\hat{\sigma}^{2}_0$ under the sparse design and $\mu_2$}
	\begin{center}
		\begin{tabular}{|c|c|c|c|c|c|}
			\cline{4-6}  \cline{5-6} \cline{6-6}
			\multicolumn{3}{c|}{} & \multicolumn{3}{c|}{method}\tabularnewline
			\hline
			Cov & $n$ & $\sigma_0^2$ & SNPT & PACE & LM\tabularnewline
			\hline
			\hline
			\multirow{8}{*}{I} & \multirow{4}{*}{50} & 0  & 0.011 (0.008) & 0.138 (0.150) & 0.146 (0.203) \tabularnewline
			\cline{3-6}
			& & 0.1  & 0.024 (0.030) & 0.144 (0.176) & 0.138 (0.145) \tabularnewline
			\cline{3-6}
			& & 0.25  & 0.062 (0.078) & 0.162 (0.173) & 0.129 (0.132) \tabularnewline
			\cline{3-6}
			& & 0.5  & 0.113 (0.145) & 0.173 (0.208) & 0.135 (0.132) \tabularnewline
			\cline{2-6}
			& \multirow{4}{*}{200} & 0  & 0.009 (0.005) & 0.080 (0.100) & 0.069 (0.076) \tabularnewline
			\cline{3-6}
			& & 0.1  & 0.014 (0.018) & 0.091 (0.098) & 0.144 (0.154) \tabularnewline
			\cline{3-6}
			& & 0.25  & 0.028 (0.033) & 0.083 (0.090) & 0.107 (0.105) \tabularnewline
			\cline{3-6}
			& & 0.5  & 0.052 (0.063) & 0.095 (0.107) & 0.139 (0.101) \tabularnewline
			\hline
			\multirow{8}{*}{II} & \multirow{4}{*}{50} & 0  & 0.036 (0.029) & 0.273 (0.272) & 0.223 (0.247) \tabularnewline
			\cline{3-6}
			& & 0.1  & 0.043 (0.049) & 0.238 (0.230) & 0.236 (0.250) \tabularnewline
			\cline{3-6}
			& & 0.25  & 0.078 (0.107) & 0.246 (0.276) & 0.168 (0.193) \tabularnewline
			\cline{3-6}
			& & 0.5  & 0.124 (0.152) & 0.257 (0.279) & 0.114 (0.131) \tabularnewline
			\cline{2-6}
			& \multirow{4}{*}{200} & 0  & 0.024 (0.015) & 0.182 (0.186) & 0.171 (0.164) \tabularnewline
			\cline{3-6}
			& & 0.1  & 0.034 (0.035) & 0.190 (0.184) & 0.188 (0.193) \tabularnewline
			\cline{3-6}
			& & 0.25  & 0.044 (0.053) & 0.184 (0.179) & 0.143 (0.141) \tabularnewline
			\cline{3-6}
			& & 0.5  & 0.067 (0.078) & 0.170 (0.170) & 0.107 (0.081) \tabularnewline
			\hline
			\multirow{8}{*}{III} & \multirow{4}{*}{50} & 0  & 0.003 (0.003) & 0.099 (0.107) & 0.014 (0.030) \tabularnewline
			\cline{3-6}
			& & 0.1  & 0.022 (0.027) & 0.098 (0.107) & 0.114 (0.152) \tabularnewline
			\cline{3-6}
			& & 0.25  & 0.058 (0.075) & 0.121 (0.130) & 0.105 (0.102) \tabularnewline
			\cline{3-6}
			& & 0.5  & 0.092 (0.121) & 0.109 (0.138) & 0.157 (0.127) \tabularnewline
			\cline{2-6}
			& \multirow{4}{*}{200} & 0  & 0.002 (0.001) & 0.063 (0.070) & 0.005 (0.009) \tabularnewline
			\cline{3-6}
			& & 0.1  & 0.010 (0.013) & 0.068 (0.067) & 0.064 (0.094) \tabularnewline
			\cline{3-6}
			& & 0.25  & 0.029 (0.035) & 0.070 (0.075) & 0.078 (0.073) \tabularnewline
			\cline{3-6}
			& & 0.5  & 0.053 (0.066) & 0.070 (0.081) & 0.148 (0.090) \tabularnewline
			\hline
	\end{tabular}
\end{center}
	\label{tab:sig-sparse-2}
\end{table}

\begin{table}
	\renewcommand{\arraystretch}{0.8}
	\caption{RMSE and their standard errors for $\hat{\sigma}^{2}_0$ under the dense design and $\mu_1$}
	\begin{center}\begin{tabular}{|c|c|c|c|c|c|}
			\cline{4-6}  \cline{5-6} \cline{6-6}
			\multicolumn{3}{c|}{} & \multicolumn{3}{c|}{method}\tabularnewline
			\hline
			Cov & $n$ & $\sigma_0^2$ & SNPT & PACE & LM\tabularnewline
			\hline
			\hline
			\multirow{8}{*}{I} & \multirow{4}{*}{50} & 0  & 0.013 (0.003) & 0.052 (0.053) & 0.036 (0.098) \tabularnewline
			\cline{3-6}
			& & 0.1  & 0.014 (0.012) & 0.050 (0.053) & 0.037 (0.052) \tabularnewline
			\cline{3-6}
			& & 0.25  & 0.019 (0.021) & 0.045 (0.044) & 0.039 (0.056) \tabularnewline
			\cline{3-6}
			& & 0.5  & 0.027 (0.032) & 0.043 (0.053) & 0.032 (0.034) \tabularnewline
			\cline{2-6}
			& \multirow{4}{*}{200} & 0  & 0.013 (0.002) & 0.043 (0.035) & 0.021 (0.008) \tabularnewline
			\cline{3-6}
			& & 0.1  & 0.013 (0.008) & 0.039 (0.031) & 0.024 (0.049) \tabularnewline
			\cline{3-6}
			& & 0.25  & 0.015 (0.013) & 0.036 (0.032) & 0.019 (0.025) \tabularnewline
			\cline{3-6}
			& & 0.5  & 0.018 (0.020) & 0.027 (0.027) & 0.015 (0.016) \tabularnewline
			\hline
			\multirow{8}{*}{II} & \multirow{4}{*}{50} & 0  & 0.025 (0.009) & 0.172 (0.122) & 0.075 (0.043) \tabularnewline
			\cline{3-6}
			& & 0.1  & 0.025 (0.016) & 0.177 (0.130) & 0.086 (0.100) \tabularnewline
			\cline{3-6}
			& & 0.25  & 0.030 (0.029) & 0.161 (0.115) & 0.074 (0.061) \tabularnewline
			\cline{3-6}
			& & 0.5  & 0.037 (0.042) & 0.155 (0.127) & 0.060 (0.059) \tabularnewline
			\cline{2-6}
			& \multirow{4}{*}{200} & 0  & 0.026 (0.006) & 0.165 (0.084) & 0.091 (0.042) \tabularnewline
			\cline{3-6}
			& & 0.1  & 0.025 (0.012) & 0.159 (0.078) & 0.091 (0.043) \tabularnewline
			\cline{3-6}
			& & 0.25  & 0.027 (0.020) & 0.152 (0.079) & 0.092 (0.048) \tabularnewline
			\cline{3-6}
			& & 0.5  & 0.030 (0.028) & 0.141 (0.083) & 0.081 (0.053) \tabularnewline
			\hline
			\multirow{8}{*}{III} & \multirow{4}{*}{50} & 0  & 0.003 (0.002) & 0.057 (0.048) & 0.002 (0.001) \tabularnewline
			\cline{3-6}
			& & 0.1  & 0.007 (0.007) & 0.051 (0.047) & 0.041 (0.078) \tabularnewline
			\cline{3-6}
			& & 0.25  & 0.014 (0.016) & 0.046 (0.047) & 0.034 (0.047) \tabularnewline
			\cline{3-6}
			& & 0.5  & 0.026 (0.033) & 0.051 (0.055) & 0.045 (0.043) \tabularnewline
			\cline{2-6}
			& \multirow{4}{*}{200} & 0  & 0.004 (0.002) & 0.053 (0.037) & 0.003 (0.002) \tabularnewline
			\cline{3-6}
			& & 0.1  & 0.006 (0.005) & 0.052 (0.034) & 0.004 (0.005) \tabularnewline
			\cline{3-6}
			& & 0.25  & 0.009 (0.010) & 0.049 (0.037) & 0.010 (0.011) \tabularnewline
			\cline{3-6}
			& & 0.5  & 0.014 (0.015) & 0.037 (0.032) & 0.022 (0.021) \tabularnewline
			\hline
	\end{tabular}\end{center}
	\label{tab:sig-dense-1}
\end{table}

\begin{table}
	\renewcommand{\arraystretch}{0.8}
	\caption{RMSE and their standard errors for $\hat{\sigma}^{2}_0$ under the dense design and $\mu_2$}
	\begin{center}\begin{tabular}{|c|c|c|c|c|c|}
			\cline{4-6}  \cline{5-6} \cline{6-6}
			\multicolumn{3}{c|}{} & \multicolumn{3}{c|}{method}\tabularnewline
			\hline
			Cov & $n$ & $\sigma_0^2$ & SNPT & PACE & LM\tabularnewline
			\hline
			\hline
			\multirow{8}{*}{I} & \multirow{4}{*}{50} & 0  & 0.012 (0.003) & 0.054 (0.050) & 0.019 (0.012) \tabularnewline
			\cline{3-6}
			& & 0.1  & 0.012 (0.011) & 0.048 (0.050) & 0.049 (0.056) \tabularnewline
			\cline{3-6}
			& & 0.25  & 0.017 (0.020) & 0.043 (0.044) & 0.058 (0.067) \tabularnewline
			\cline{3-6}
			& & 0.5  & 0.027 (0.033) & 0.043 (0.049) & 0.034 (0.039) \tabularnewline
			\cline{2-6}
			& \multirow{4}{*}{200} & 0  & 0.012 (0.003) & 0.045 (0.034) & 0.021 (0.008) \tabularnewline
			\cline{3-6}
			& & 0.1  & 0.012 (0.007) & 0.040 (0.033) & 0.031 (0.060) \tabularnewline
			\cline{3-6}
			& & 0.25  & 0.014 (0.013) & 0.034 (0.029) & 0.030 (0.045) \tabularnewline
			\cline{3-6}
			& & 0.5  & 0.017 (0.019) & 0.029 (0.031) & 0.016 (0.017) \tabularnewline
			\hline
			\multirow{8}{*}{II} & \multirow{4}{*}{50} & 0  & 0.024 (0.008) & 0.168 (0.118) & 0.074 (0.042) \tabularnewline
			\cline{3-6}
			& & 0.1  & 0.025 (0.015) & 0.160 (0.123) & 0.082 (0.086) \tabularnewline
			\cline{3-6}
			& & 0.25  & 0.030 (0.027) & 0.161 (0.124) & 0.076 (0.057) \tabularnewline
			\cline{3-6}
			& & 0.5  & 0.037 (0.037) & 0.146 (0.114) & 0.059 (0.057) \tabularnewline
			\cline{2-6}
			& \multirow{4}{*}{200} & 0  & 0.025 (0.006) & 0.163 (0.084) & 0.090 (0.042) \tabularnewline
			\cline{3-6}
			& & 0.1  & 0.025 (0.012) & 0.161 (0.086) & 0.095 (0.043) \tabularnewline
			\cline{3-6}
			& & 0.25  & 0.026 (0.019) & 0.155 (0.084) & 0.092 (0.047) \tabularnewline
			\cline{3-6}
			& & 0.5  & 0.027 (0.027) & 0.142 (0.086) & 0.080 (0.049) \tabularnewline
			\hline
			\multirow{8}{*}{III} & \multirow{4}{*}{50} & 0  & 0.002 (0.002) & 0.057 (0.049) & 0.001 (0.001) \tabularnewline
			\cline{3-6}
			& & 0.1  & 0.007 (0.007) & 0.053 (0.049) & 0.045 (0.054) \tabularnewline
			\cline{3-6}
			& & 0.25  & 0.013 (0.014) & 0.049 (0.050) & 0.042 (0.063) \tabularnewline
			\cline{3-6}
			& & 0.5  & 0.025 (0.031) & 0.050 (0.054) & 0.041 (0.041) \tabularnewline
			\cline{2-6}
			& \multirow{4}{*}{200} & 0  & 0.003 (0.002) & 0.055 (0.034) & 0.003 (0.003) \tabularnewline
			\cline{3-6}
			& & 0.1  & 0.005 (0.005) & 0.054 (0.037) & 0.004 (0.005) \tabularnewline
			\cline{3-6}
			& & 0.25  & 0.006 (0.007) & 0.046 (0.034) & 0.008 (0.010) \tabularnewline
			\cline{3-6}
			& & 0.5  & 0.012 (0.015) & 0.036 (0.034) & 0.023 (0.021) \tabularnewline
			\hline
	\end{tabular}\end{center}
	\label{tab:sig-dense-2}
\end{table}

\newpage
\section{Additional simulation results for $\sigma_X^2(t)$}
\begin{table}[h!]
	\caption{RMISE and their standard errors for $\hat\sigma_X^2(t)$ under the sparse design and $\mu_2$}
	\begin{center}
		\begin{tabular}{|c|c|c|c|c|c|}
			\cline{4-6}  \cline{5-6} \cline{6-6}
			\multicolumn{3}{c|}{} & \multicolumn{3}{c|}{method}\tabularnewline
			\hline
			Cov & SNR & $n$ & SNPTM & PFBE & PACE\tabularnewline
			\hline
			\hline
			\multirow{4}{*}{I} & \multirow{2}{*}{2} & 50  & 0.523 (0.206) & 0.513 (0.222) & 2.012 (1.296) \tabularnewline
			\cline{3-6}
			& & 200  & 0.330 (0.118) & 0.319 (0.114) & 1.274 (0.880) \tabularnewline
			\cline{2-6}
			& \multirow{2}{*}{4} & 50  & 0.517 (0.219) & 0.480 (0.230) & 1.653 (1.261) \tabularnewline
			\cline{3-6}
			& & 200  & 0.315 (0.137) & 0.321 (0.136) & 1.213 (0.798) \tabularnewline
			\hline
			\multirow{4}{*}{II} & \multirow{2}{*}{2} & 50  & 0.759 (0.281) & 0.731 (0.211) & 2.521 (1.546) \tabularnewline
			\cline{3-6}
			& & 200  & 0.495 (0.165) & 0.520 (0.156) & 1.890 (1.382) \tabularnewline
			\cline{2-6}
			& \multirow{2}{*}{4} & 50  & 0.756 (0.261) & 0.726 (0.246) & 2.236 (1.511) \tabularnewline
			\cline{3-6}
			& & 200  & 0.468 (0.178) & 0.492 (0.142) & 1.376 (0.982) \tabularnewline
			\hline
			\multirow{4}{*}{III} & \multirow{2}{*}{2} & 50  & 0.580 (0.178) & 0.555 (0.138) & 1.598 (1.181) \tabularnewline
			\cline{3-6}
			& & 200  & 0.399 (0.234) & 0.412 (0.120) & 0.995 (0.581) \tabularnewline
			\cline{2-6}
			& \multirow{2}{*}{4} & 50  & 0.550 (0.183) & 0.550 (0.126) & 1.089 (0.885) \tabularnewline
			\cline{3-6}
			& & 200  & 0.354 (0.206) & 0.386 (0.145) & 0.953 (0.555) \tabularnewline
			\hline
		\end{tabular}
	\end{center}
	\label{tab:var-sparse-2}
\end{table}

\begin{table}
	\caption{RMISE and their standard errors for $\hat\sigma_X^2(t)$ under the dense design and $\mu_1$}
	\begin{center}\begin{tabular}{|c|c|c|c|c|c|}
			\cline{4-6}  \cline{5-6} \cline{6-6}
			\multicolumn{3}{c|}{} & \multicolumn{3}{c|}{method}\tabularnewline
			\hline
			Cov & SNR & $n$ & SNPTM & PFBE & PACE\tabularnewline
			\hline
			\hline
			\multirow{4}{*}{I} & \multirow{2}{*}{2} & 50  & 0.488 (0.117) & 0.509 (0.245) & 0.588 (0.246) \tabularnewline
			\cline{3-6}
			& & 200  & 0.274 (0.082) & 0.275 (0.087) & 0.344 (0.109) \tabularnewline
			\cline{2-6}
			& \multirow{2}{*}{4} & 50  & 0.480 (0.115) & 0.502 (0.198) & 0.561 (0.221) \tabularnewline
			\cline{3-6}
			& & 200  & 0.264 (0.071) & 0.266 (0.077) & 0.331 (0.094) \tabularnewline
			\hline
			\multirow{4}{*}{II} & \multirow{2}{*}{2} & 50  & 0.665 (0.157) & 0.676 (0.202) & 0.757 (0.254) \tabularnewline
			\cline{3-6}
			& & 200  & 0.393 (0.139) & 0.414 (0.108) & 0.526 (0.125) \tabularnewline
			\cline{2-6}
			& \multirow{2}{*}{4} & 50  & 0.657 (0.147) & 0.649 (0.214) & 0.747 (0.308) \tabularnewline
			\cline{3-6}
			& & 200  & 0.362 (0.118) & 0.396 (0.120) & 0.493 (0.120) \tabularnewline
			\hline
			\multirow{4}{*}{III} & \multirow{2}{*}{2} & 50  & 0.504 (0.125) & 0.490 (0.960) & 0.793 (0.311) \tabularnewline
			\cline{3-6}
			& & 200  & 0.297 (0.096) & 0.238 (0.074) & 0.797 (0.184) \tabularnewline
			\cline{2-6}
			& \multirow{2}{*}{4} & 50  & 0.497 (0.125) & 0.414 (0.235) & 0.797 (0.297) \tabularnewline
			\cline{3-6}
			& & 200  & 0.271 (0.076) & 0.223 (0.070) & 0.786 (0.170) \tabularnewline
			\hline
		\end{tabular}
	\end{center}
	\label{tab:var-dense-1}
\end{table}

\begin{table}
	\caption{RMISE and their standard errors for $\hat\sigma_X^2(t)$ under the dense design and $\mu_2$}
	\begin{center}\begin{tabular}{|c|c|c|c|c|c|}
			\cline{4-6}  \cline{5-6} \cline{6-6}
			\multicolumn{3}{c|}{} & \multicolumn{3}{c|}{method}\tabularnewline
			\hline
			Cov & SNR & $n$ & SNPTM & PFBE & PACE\tabularnewline
			\hline
			\hline
			\multirow{4}{*}{I} & \multirow{2}{*}{2} & 50  & 0.491 (0.158) & 0.521 (0.150) & 0.573 (0.248) \tabularnewline
			\cline{3-6}
			& & 200  & 0.260 (0.083) & 0.259 (0.087) & 0.330 (0.101) \tabularnewline
			\cline{2-6}
			& \multirow{2}{*}{4} & 50  & 0.476 (0.132) & 0.511 (0.146) & 0.551 (0.259) \tabularnewline
			\cline{3-6}
			& & 200  & 0.248 (0.064) & 0.257 (0.089) & 0.322 (0.107) \tabularnewline
			\hline
			\multirow{4}{*}{II} & \multirow{2}{*}{2} & 50  & 0.676 (0.159) & 0.668 (0.193) & 0.748 (0.311) \tabularnewline
			\cline{3-6}
			& & 200  & 0.390 (0.126) & 0.412 (0.109) & 0.524 (0.134) \tabularnewline
			\cline{2-6}
			& \multirow{2}{*}{4} & 50  & 0.654 (0.201) & 0.660 (0.156) & 0.726 (0.259) \tabularnewline
			\cline{3-6}
			& & 200  & 0.371 (0.105) & 0.396 (0.110) & 0.512 (0.124) \tabularnewline
			\hline
			\multirow{4}{*}{III} & \multirow{2}{*}{2} & 50  & 0.505 (0.126) & 0.414 (0.201) & 0.839 (0.377) \tabularnewline
			\cline{3-6}
			& & 200  & 0.290 (0.113) & 0.274 (0.153) & 0.812 (0.215) \tabularnewline
			\cline{2-6}
			& \multirow{2}{*}{4} & 50  & 0.487 (0.149) & 0.398 (0.200) & 0.825 (0.367) \tabularnewline
			\cline{3-6}
			& & 200  & 0.275 (0.077) & 0.245 (0.127) & 0.787 (0.171) \tabularnewline
			\hline
		\end{tabular}
	\end{center}
	\label{tab:var-dense-2}
\end{table}

\newpage

\section{Additional simulation results for $\mathcal{C}$}
\begin{table}[h!]
	\caption{RMISE and their standard errors for $\hat{\mathcal{C}}$ under the sparse design and $\mu_2$}
	\begin{center}\begin{tabular}{|c|c|c|c|c|c|c|}
			\cline{4-7}  \cline{5-7} \cline{6-7} \cline{7-7} 
			\multicolumn{3}{c|}{} & \multicolumn{4}{c|}{method}\tabularnewline
			\hline
			Cov & SNR & $n$ & SNPTM & SNPTF & PFBE & PACE\tabularnewline
			\hline
			\hline
			\multirow{4}{*}{I} & \multirow{2}{*}{2} & 50  & 0.338 (0.111) & 0.443 (0.148) & 0.488 (0.169) & 1.416 (0.673) \tabularnewline
			\cline{3-7}
			& & 200  & 0.237 (0.089) & 0.354 (0.087) & 0.296 (0.082) & 1.015 (0.506) \tabularnewline
			\cline{2-7}
			& \multirow{2}{*}{4} & 50  & 0.319 (0.137) & 0.418 (0.133) & 0.421 (0.162) & 1.269 (0.684) \tabularnewline
			\cline{3-7}
			& & 200  & 0.221 (0.087) & 0.337 (0.084) & 0.271 (0.110) & 0.966 (0.409) \tabularnewline
			\hline
			\multirow{4}{*}{II} & \multirow{2}{*}{2} & 50  & 0.567 (0.124) & 0.519 (0.174) & 0.545 (0.163) & 2.259 (1.290) \tabularnewline
			\cline{3-7}
			& & 200  & 0.481 (0.074) & 0.428 (0.132) & 0.468 (0.114) & 1.706 (0.747) \tabularnewline
			\cline{2-7}
			& \multirow{2}{*}{4} & 50  & 0.542 (0.129) & 0.463 (0.137) & 0.510 (0.176) & 1.973 (1.167) \tabularnewline
			\cline{3-7}
			& & 200  & 0.466 (0.059) & 0.413 (0.110) & 0.432 (0.098) & 1.471 (0.607) \tabularnewline
			\hline
			\multirow{4}{*}{III} & \multirow{2}{*}{2} & 50  & 0.498 (0.085) & 0.492 (0.124) & 0.483 (0.104) & 1.305 (0.721) \tabularnewline
			\cline{3-7}
			& & 200  & 0.479 (0.049) & 0.446 (0.092) & 0.371 (0.061) & 1.137 (0.404) \tabularnewline
			\cline{2-7}
			& \multirow{2}{*}{4} & 50  & 0.485 (0.073) & 0.479 (0.122) & 0.477 (0.104) & 1.217 (0.603) \tabularnewline
			\cline{3-7}
			& & 200  & 0.473 (0.041) & 0.433 (0.083) & 0.363 (0.080) & 1.070 (0.373) \tabularnewline
			\hline
		\end{tabular}
	\end{center}
	\label{tab:cov-sparse-2}
\end{table}

\begin{table}
	\caption{RMISE and their standard errors for $\hat{\mathcal{C}}$ under the dense design and $\mu_1$}
	\begin{center}\begin{tabular}{|c|c|c|c|c|c|c|}
			\cline{4-7}  \cline{5-7} \cline{6-7} \cline{7-7} 
			\multicolumn{3}{c|}{} & \multicolumn{4}{c|}{method}\tabularnewline
			\hline
			Cov & SNR & $n$ & SNPTM & SNPTF & PFBE & PACE\tabularnewline
			\hline
			\hline
			\multirow{4}{*}{I} & \multirow{2}{*}{2} & 50  & 0.308 (0.096) & 0.412 (0.101) & 0.434 (0.132) & 0.582 (0.171) \tabularnewline
			\cline{3-7}
			& & 200  & 0.177 (0.064) & 0.301 (0.064) & 0.260 (0.077) & 0.469 (0.079) \tabularnewline
			\cline{2-7}
			& \multirow{2}{*}{4} & 50  & 0.289 (0.089) & 0.402 (0.079) & 0.423 (0.117) & 0.542 (0.128) \tabularnewline
			\cline{3-7}
			& & 200  & 0.169 (0.052) & 0.284 (0.067) & 0.250 (0.068) & 0.449 (0.078) \tabularnewline
			\hline
			\multirow{4}{*}{II} & \multirow{2}{*}{2} & 50  & 0.528 (0.085) & 0.489 (0.116) & 0.516 (0.154) & 1.069 (0.324) \tabularnewline
			\cline{3-7}
			& & 200  & 0.397 (0.042) & 0.351 (0.085) & 0.382 (0.148) & 1.000 (0.213) \tabularnewline
			\cline{2-7}
			& \multirow{2}{*}{4} & 50  & 0.502 (0.090) & 0.465 (0.119) & 0.499 (0.175) & 1.045 (0.277) \tabularnewline
			\cline{3-7}
			& & 200  & 0.382 (0.035) & 0.341 (0.074) & 0.371 (0.134) & 0.981 (0.207) \tabularnewline
			\hline
			\multirow{4}{*}{III} & \multirow{2}{*}{2} & 50  & 0.471 (0.044) & 0.453 (0.081) & 0.454 (0.146) & 0.949 (0.217) \tabularnewline
			\cline{3-7}
			& & 200  & 0.467 (0.027) & 0.427 (0.052) & 0.336 (0.046) & 1.015 (0.130) \tabularnewline
			\cline{2-7}
			& \multirow{2}{*}{4} & 50  & 0.463 (0.039) & 0.422 (0.085) & 0.393 (0.127) & 0.961 (0.221) \tabularnewline
			\cline{3-7}
			& & 200  & 0.463 (0.026) & 0.402 (0.046) & 0.337 (0.039) & 1.010 (0.130) \tabularnewline
			\hline
		\end{tabular}
	\end{center}
	\label{tab:cov-dense-1}
\end{table}

\begin{table}
	\caption{RMISE and their standard errors for $\hat{\mathcal{C}}$ under the dense design and $\mu_2$}
	\begin{center}\begin{tabular}{|c|c|c|c|c|c|c|}
			\cline{4-7}  \cline{5-7} \cline{6-7} \cline{7-7} 
			\multicolumn{3}{c|}{} & \multicolumn{4}{c|}{method}\tabularnewline
			\hline
			Cov & SNR & $n$ & SNPTM & SNPTF & PFBE & PACE\tabularnewline
			\hline
			\hline
			\multirow{4}{*}{I} & \multirow{2}{*}{2} & 50  & 0.305 (0.104) & 0.405 (0.100) & 0.415 (0.140) & 0.590 (0.146) \tabularnewline
			\cline{3-7}
			& & 200  & 0.173 (0.062) & 0.295 (0.070) & 0.252 (0.072) & 0.477 (0.081) \tabularnewline
			\cline{2-7}
			& \multirow{2}{*}{4} & 50  & 0.294 (0.092) & 0.401 (0.098) & 0.400 (0.138) & 0.570 (0.159) \tabularnewline
			\cline{3-7}
			& & 200  & 0.163 (0.046) & 0.278 (0.054) & 0.238 (0.067) & 0.449 (0.082) \tabularnewline
			\hline
			\multirow{4}{*}{II} & \multirow{2}{*}{2} & 50  & 0.533 (0.086) & 0.519 (0.113) & 0.526 (0.159) & 1.048 (0.305) \tabularnewline
			\cline{3-7}
			& & 200  & 0.420 (0.041) & 0.352 (0.071) & 0.393 (0.154) & 1.005 (0.233) \tabularnewline
			\cline{2-7}
			& \multirow{2}{*}{4} & 50  & 0.523 (0.092) & 0.489 (0.116) & 0.502 (0.186) & 1.023 (0.296) \tabularnewline
			\cline{3-7}
			& & 200  & 0.392 (0.033) & 0.346 (0.065) & 0.378 (0.147) & 0.970 (0.194) \tabularnewline
			\hline
			\multirow{4}{*}{III} & \multirow{2}{*}{2} & 50  & 0.479 (0.047) & 0.446 (0.096) & 0.395 (0.089) & 0.977 (0.277) \tabularnewline
			\cline{3-7}
			& & 200  & 0.464 (0.027) & 0.420 (0.051) & 0.351 (0.152) & 0.926 (0.151) \tabularnewline
			\cline{2-7}
			& \multirow{2}{*}{4} & 50  & 0.482 (0.049) & 0.427 (0.090) & 0.393 (0.085) & 0.968 (0.245) \tabularnewline
			\cline{3-7}
			& & 200  & 0.461 (0.026) & 0.409 (0.049) & 0.337 (0.066) & 0.916 (0.116) \tabularnewline
			\hline
		\end{tabular}
	\end{center}
	\label{tab:cov-dense-2}
\end{table}

\end{document}